\documentclass[research]{fcs}

\usepackage{graphicx}

\usepackage{amsmath,bm}
\usepackage{amssymb}
\usepackage{comment}
\usepackage{xcolor}
\usepackage{algpseudocode}
\usepackage{xspace}
\usepackage{subfigure}
\usepackage{booktabs}
\DeclareMathOperator*{\argmax}{arg\,max}

\usepackage{ragged2e}
\usepackage[linesnumbered,ruled,vlined,boxed,noend]{algorithm2e}
\usepackage{url}
\newtheorem{definition1}{Definition}
\newtheorem{example1}{Example}
\newlist{inlinelist}{enumerate*}{1}
\setlist*[inlinelist,1]{
	label=(\roman*),
}
\usepackage[multiple]{footmisc}
\newcommand{\myparagraph}[1]{\vspace{0.3\baselineskip}\noindent{\textbf{#1.}}~}
\usepackage{multirow} 
\newcommand{\CAMC}{\textsf{CMC+MG}\xspace}
\newcommand{\CAMS}{\textsf{CMC+MC}\xspace}
\newcommand{\BGPA}{\textsf{DPSA}\xspace}
\newcommand{\BGPAF}{\textsf{DPSA+BA}\xspace}
\newcommand{\SG}{\textsf{DSA}\xspace}

\newcommand{\MCPBC}{\textbf{BMCC}\xspace} 
\newcommand{\SpatialSet}{cell-based dataset\xspace}
\newcommand{\SpatialSets}{cell-based datasets\xspace}
\newcommand{\market}{\mathcal{S}_\mathcal{D}\xspace}

\usepackage{soul}
\usepackage{etoolbox}

\usepackage{color,soul}
\definecolor{Maroon}{rgb}{0.678, 0.090, 0.216}
\definecolor{Brown}{rgb}{0.59, 0.29, 0.0}
\definecolor{Green}{rgb}{0,1,0}
\definecolor{NavyBlue}{rgb}{0.0, 0.0, 0.5}

\setul{0.5ex}{0.3ex}
\definecolor{Green}{rgb}{0,1,0}
\setulcolor{Green}

\title{Budgeted Spatial Data Acquisition: When Coverage and Connectivity Matter}

\author[1]{Wenzhe YANG}
\author[2]{Shixun HUANG}
\author*[1]{Sheng WANG}
\author*[1]{Zhiyong PENG}
\address[1]{School of Computer Science, Wuhan University, Wuhan 430061, China}
\address[2]{School of Computing and Information Technology,  The University of Wollongong, Wollongong 2522,  Australia}

\fcssetup{
  received       = {month dd, yyyy},
  accepted       = {month dd, yyyy},
  corr-email     = {swangcs@whu.edu.cn, peng@whu.edu.cn},
}
\begin{abstract}
Data is undoubtedly becoming a commodity like oil, land, and labour in the 21st century. Although there have been many successful marketplaces for data trading, the existing data marketplaces lack consideration of the case where buyers want to acquire a collection of datasets (instead of one), and the overall spatial coverage and connectivity matter. In this paper, we make the first attempt to formulate this problem as Budgeted Maximum Coverage with Connectivity Constraint (\MCPBC), which aims to acquire a dataset collection with the maximum spatial coverage under a limited budget while maintaining spatial connectivity. To solve the problem, we propose two approximate algorithms with detailed theoretical guarantees and time complexity analysis, followed by two acceleration strategies to further improve the efficiency of the algorithm. Experiments are conducted on five real-world spatial dataset collections to verify the efficiency and effectiveness of our algorithms.
\end{abstract}

\begin{document}

\section{Introduction}
\label{sec:intro}
\begin{sloppypar}
Data is often referred to as the new oil of the digital economy, representing a highly valuable and untapped asset~\cite{ChapmanSKKIKG20,Fernandez2020}. To fully realize the potential of data, various data marketplace platforms have emerged and offer a diverse range of data types for trade. For instance, Xignite sells financial data, Gnip sells data from social media, and Datarade sells multiple categories of data, including commercial, academic, and spatial data, etc. Among them, spatial data has become increasingly important. According to a research report by Verified Market, the global spatial analytics market was valued at \$66 billion in 2021 and is expected to reach \$209 billion by the end of 2030.

\begin{figure*}
\setlength{\abovecaptionskip}{0.3 cm}
\setlength{\belowcaptionskip}{-0.3 cm}	
\centering
\includegraphics[width = 0.9\textwidth]{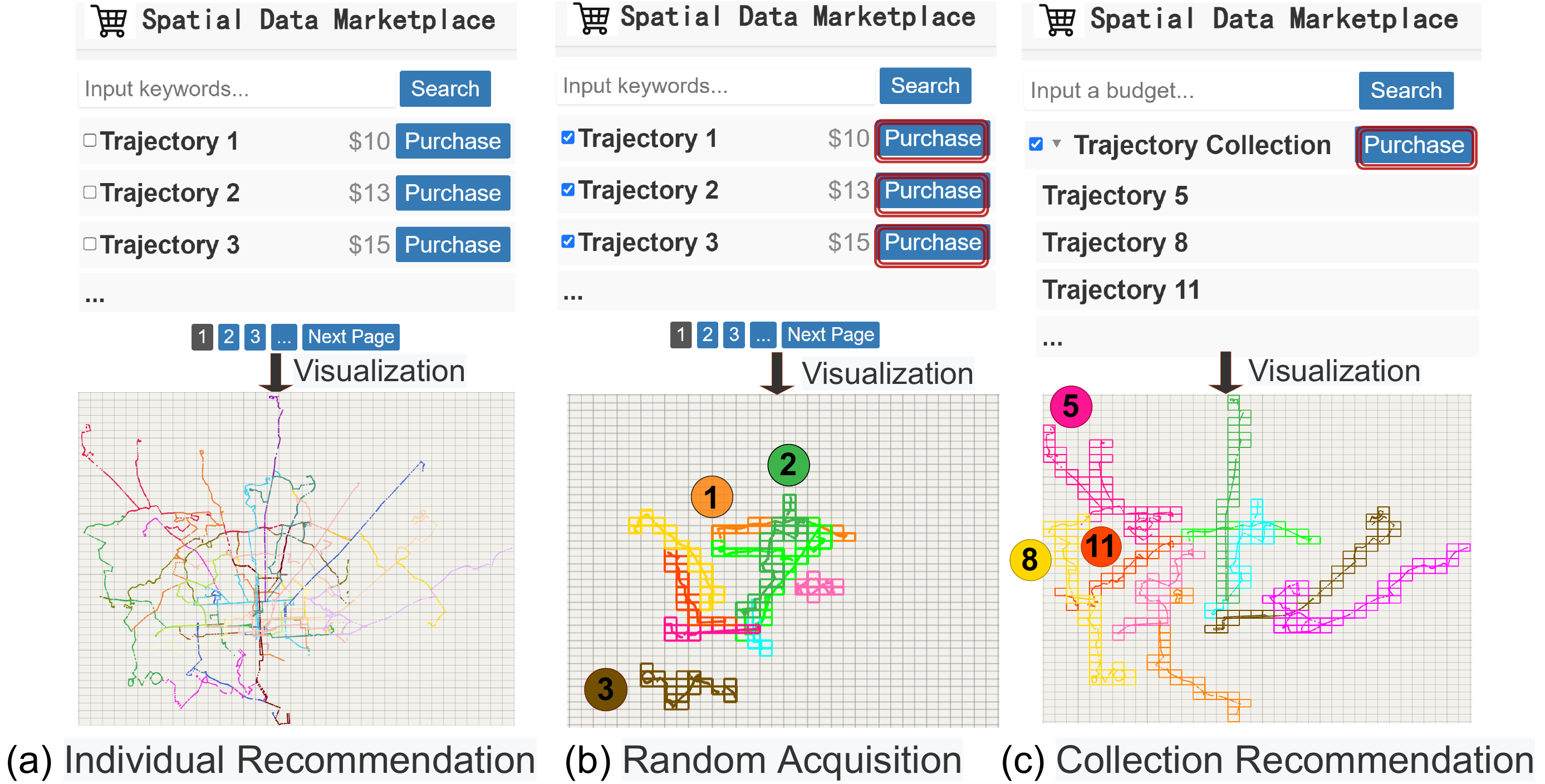}
  \caption{
  An example of dataset acquisition under a limited budget. (a) shows the individual dataset recommendation in the current marketplaces, (b) forms the collection by randomly choosing individual recommended datasets, whereas (c) directly recommends the dataset collection. 
  }
  \label{fig:intro} 
  \end{figure*}

Driven by the enormous economic profit, the geospatial data marketplaces, where buyers acquire or purchase data provided by sellers, are experiencing rapid growth. For example, the EU-funded OpertusMundi project develops the Topio~\cite{Ionescu2023} to sell geospatial datasets of Europe, the UP42 sells a wide range of earth observation dataset, and AggData sells various location datasets such as restaurants and stations.

\noindent \textbf{The Gap in Marketplace - Individual Acquisition}. However, despite notable efforts in spatial data acquisition, current marketplaces primarily focus on recommending each dataset individually, as shown in Fig.~\ref{fig:intro}(a). There is a lack of consideration for cases where an individual dataset cannot satisfy the buyer's needs such that a collection of datasets needs to be acquired. For instance, in data-driven transit planning~\cite{wang2021public,WangTKDE2018,Wang2020,Ali2018,HeZZK22}, planners require a collection of commuting trajectories to analyze the demand of citizens and design new routes. When it comes to spatial dataset collection recommendation, we need to consider two important properties of the collection, namely \emph{spatial coverage} and \emph{spatial connectivity}, to ensure that buyers can make the best use of a limited budget. In the following, we will use an example to show the significance of a good dataset collection recommendation and the high-level idea of these two properties. Afterwards, we will describe these two properties more formally, followed by our problem formulation and solution description.

\begin{example}

Consider a scenario where the data buyer is a transportation company manager under a limited budget. The buyer wants to acquire a collection of private trajectories uploaded and sold by sellers, which are used to know the citizens' demands and redesign new bus routes. Fig.~\ref{fig:intro}(a) simply shows that the existing data market usually recommends trajectories individually for buyers to choose from, where each trajectory comprises a set of closely distributed points with the same color. 
The data buyers can acquire a collection of data through random purchases, and the results are shown in Fig.~\ref{fig:intro}(b). In this figure, the colored cells represent the spatial coverage of a trajectory. 

It is evident that these datasets are distributed over a small area, and exhibit notable overlap (e.g., the trajectories $\text{\textcircled{1}}$ and $\text{\textcircled{2}}$ of Fig.~\ref{fig:intro}(b)) or are separated (e.g., the trajectory $\text{\textcircled{3}}$ of Fig.~\ref{fig:intro}(b)). These results make it difficult to design routes that cover the entire region, and separate datasets affect the transfer experience of new bus routes. Fig.~\ref{fig:intro}(c) shows a better recommendation result. The datasets acquired are widely distributed throughout the region, and there is a connection between routes (e.g., the trajectories $\text{\textcircled{5}}$ $\text{\textcircled{8}}$ and $\text{\textcircled{8}}$ $\text{\textcircled{11}}$ of Fig.~\ref{fig:intro}(c)), which facilitates redesigning new bus routes since they cover a large area and ensure the convenience of transfers.
\end{example}

\noindent \textbf{The Intuition of Spatial Coverage and Connectivity}. Motivated by the above application, we make the first attempt to define a combinatorial optimization problem where two characteristics, namely \emph{spatial coverage} and \emph{spatial connectivity} of acquired datasets, can play critical roles in the solution quality. In what follows, we will describe what spatial coverage and spatial connectivity are. The spatial dataset is represented as a set of points with geographic coordinates, making it challenging to quantify its coverage due to the point having no physical dimensions (length and width). To quantify the coverage of a spatial dataset, we partition the space into uniform cells, as depicted by the gray cells in Fig.~\ref{fig:intro}(a). The spatial dataset coverage is defined as the number of cells containing these spatial points. For instance, the spatial coverage of the trajectory $\text{\textcircled{3}}$ in Fig.~\ref{fig:intro}(b) is the number of cells marked with brown rectangles.

Additionally, to illustrate the spatial connectivity, it is essential to employ a metric that indicates the relationship between two datasets. As shown in Fig.~\ref{fig:intro}(c), two datasets are directly connected when they have overlapped cells (e.g., the trajectories $\text{\textcircled{5}}$ and $\text{\textcircled{8}}$ of Fig.~\ref{fig:intro}(c)). On the other hand, two datasets are deemed indirectly connected when they can establish a connection through one or multiple intermediary datasets (e.g., the trajectories $\text{\textcircled{5}}$ and $\text{\textcircled{11}}$ are indirectly connected in Fig.~\ref{fig:intro}(c) since both of them are directly connected to $\text{\textcircled{8}}$). 
The spatial connectivity requires that any pair of datasets in the result be directly or indirectly connected, which is crucial for many applications. For instance, in the aforementioned example, it can ensure that new bus routes are connected to each other, providing more transfer options.

\noindent \textbf{The Problem We Study - Collection Acquisition}. In this paper, we take the first attempt to formulate the data collection acquisition problem as the \textbf{B}udgeted \textbf{M}aximum \textbf{C}overage with \textbf{C}onnectivity Constraint (\textbf{BMCC}), which aims to find a set of datasets with the maximum spatial coverage under a limited budget while maintaining spatial connectivity. To solve the maximization problem under multiple constraints at the same time, we first transform it into a graph theory problem. Specifically, we construct a graph to model dataset relationships, where each node represents a dataset, and each edge represents that two datasets are connected. The objective of \textbf{BMCC} is to find a set of nodes (the total price within the budget) that induces a connected subgraph to maximize spatial coverage.

\begin{sloppypar}
The difficulty of the \textbf{BMCC} problem arises from its NP-hard nature, which refers to a problem for which no known polynomial-time algorithm can find a solution. We formally prove that \textbf{BMCC} is NP-hard by using a polynomial-time reduction from the classic maximum coverage problem (MCP)~\cite{Hochbaum1998} in Section~\ref{sec:nphard}. 
Inspired by the greedy algorithm proposed by Fiege~\cite{Feige1998}, which
is the best possible polynomial time approximation algorithm for the maximum coverage problem and its variants.
Thus, we propose two heuristic algorithms based on greedy strategy to solve the \textbf{BMCC} problem. The intuition behind these algorithms is to iteratively select the best available candidate dataset that maximizes the coverage, while considering the constraints imposed by the problem.

Specifically, we propose two approximate algorithms with theoretical guarantees, namely Dual-Search Algorithm (\SG) and Dual-Path Search Algorithm (\BGPA), followed by two acceleration strategies. The basic idea of \SG is to iteratively pick the dataset, which brings the maximum marginal gain w.r.t. the spatial coverage while maintaining the spatial connectivity. However, the theoretical analysis shows that the approximation ratio of \textsf{DSA} gradually decreases as the budget increases.
To address this, we propose \BGPA, which iteratively selects paths (i.e., a sequence of nodes connected by edges) with the common node and the maximum marginal gain from the dataset graph. The theoretical analysis shows its better approximation in scenarios involving a larger budget. Our contributions are summarized below:
\end{sloppypar}

\begin{itemize}[leftmargin=*]
     \item We make the first attempt to study and formalize the \MCPBC problem, motivated by the urgent needs in spatial data marketplace to recommend dataset collection for buyers. We also prove its NP-hardness.

    \item To solve the \MCPBC problem, we first propose an approximate algorithm \SG with detailed theoretical guarantees and time complexity analysis.

    \item To further improve the approximation ratio with large input budgets, we propose \BGPA with detailed theoretical guarantees and time complexity analysis. Furthermore, we propose two acceleration strategies to enhance the efficiency of \BGPA significantly.
    
    \item We conduct extensive experiments to verify the effectiveness and efficiency of our algorithms. The experimental results show that our method can achieve up to at most 68\% times larger spatial coverage with 89\% times speedups over baselines extended from the state-of-the-art on the connected maximum coverage problem. The case studies validate the effectiveness of our algorithms.
 
\end{itemize}

The rest of this paper is organized as follows. In Section~\ref{sec:relatedWork}, we give a comprehensive review of related literature. In Section~\ref{sec:defnitions}, we set the scope of the paper by formally defining the \MCPBC problem. In Section~\ref{sec:nphard}, we prove the NP-hardness of \MCPBC by performing a
reduction from the NP-hard maximum coverage problem~\cite{Hochbaum1998}. Section~\ref{sec:algorithms} presents two extension algorithms based on the greedy strategy and proves its approximation ratio theoretically. 
Finally, we present the experimental results supporting our claims in Section~\ref{sec:exp} and conclude in Section~\ref{sec:conclude}.

\end{sloppypar}

\section{Related Work}
\label{sec:relatedWork}
\begin{sloppypar}

In this section, we give a literature review on the data marketplace~~\cite{Schomm2013,Fernandez2020,Balazinska2011,Li2018,Qiao2023,Asudeh2022,Kanza2015,Ionescu2023} and combinatorial optimization problem~\cite{2003Combinatorial,Christos1998}. The data marketplace offers a platform for data acquisition where sellers upload their datasets and buyers purchase desired datasets. Moreover, the \MCPBC studied in this paper is reminiscent of the classic maximum coverage problem~\cite{Hochbaum1998}, which is one of the typical combinatorial optimization problems. It inspires us to find a new way to acquire datasets in the spatial data marketplace.

\subsection{Data Marketplace}
The literature on data marketplace can be broadly divided into three classes: (i) data marketplace model~\cite{Agarwal2019,Moor2019,Fernandez2020}, (ii) data pricing~\cite{Niu2022,Koutris2013,Zhang2023,Chawla2019,ChenK2019,lin2014}, and (iii) data acquisition~\cite{Asudeh2022,Li2018,LiYF2021}. We will provide an overview of each of these classes.

\myparagraph{Data Marketplace Model} 
A data marketplace model~\cite{Sakr2018,Balazinska2011,Fernandez2020,Azcoitia2022,PeiF023,Agarwal2019,Moor2019} consists of buyers who want to buy datasets, sellers who want to share datasets in exchange for a reward, and an arbiter who satisfies the buyers’ requests with the seller’s supplied datasets. For instance, Agarwal et al.~\cite{Agarwal2019} propose a data model that solves the allocation and payment problem in a static scenario with one buyer and multiple sellers. In~\cite{Moor2019}, the authors propose an end-to-end market design that considers buyers and sellers arriving in a streaming fashion. Additionally, Fernandez et al.~\cite{Fernandez2020} present their vision for the design of platforms and design the market rules to govern the interactions between sellers, buyers, and an arbiter.

\myparagraph{Data Pricing}
Data pricing~\cite{ChenK2019,Chawla2019,Niu2022,Koutris2013,Zhang2023,ChenYWWL2022,Tongyx2018} has been widely studied by the communities of computer science and economics (see~\cite{Pei2020,Pei2022} for a survey). For instance, for query-based pricing~\cite{Koutris2012,Koutris2013}, Koutris et al.~\cite{Koutris2013} propose the first pricing system named QueryMarket, which supports the efficient computation of prices for multiple queries using integer linear program solvers. Deep et al.~\cite{DeepK2017} propose history-aware pricing that supports buyer charges according to current and past queries, which can support pricing in real time. For machine learning-based (ML) pricing~\cite{Liu2020DealerED}, Chen et al.~\cite{ChenK2019} first formally describe the desired properties of the model-based pricing framework and provide algorithmic solutions. In addition, there are many other types of pricing mechanisms, such as feature-based pricing~\cite{Miao2022,Yu2017}, usage-based pricing~\cite{AzcoitiaIL2023}, etc. Note that our study is orthogonal to the data pricing studies.
Our algorithm proposed in this paper can work with any pricing mechanism. For simplicity, we just use the usage-based pricing strategy~\cite{AzcoitiaIL2023} for illustration.

\myparagraph{Data Acquisition} Data acquisition is another key factor in helping data buyers obtain high-quality data from the marketplace. For example, Li et al.~\cite{Li2018} propose a middleware called DANCE, which aims to acquire a target table from the data marketplace with a high correlation between its attributes and query table. 
Chen et al.~\cite{ChenILSZ18} consider the problem of obtaining data from multiple data providers in order to improve the accuracy of linear statistical estimators. 
Li et al.~\cite{LiYF2021} investigate the problem of data acquisition by a specific query predicate (e.g., ten images with label = `dog’) and budget for improving the performance of ML models. Similarly, Asudeh et al.~\cite{Asudeh2022} develop a query-answering framework to acquire an integrating dataset that meets user-provided query predicate and distribution requirements (e.g., 100 cancer datasets with at least 20\% black patients).

Despite notable progress being made in this direction, existing studies primarily concentrate on tabular data, which typically refers to datasets organized in tables with rows and columns. In contrast to tabular data, spatial data has spatial attributes associated with each data point (e.g., latitude and longitude). Although spatial data can be converted into tabular form to some extent, this representation may not adequately reflect its spatial characteristics and complexity. Thus, these approaches for acquiring tabular datasets are not directly applicable to the acquisition of spatial datasets. This is because tabular data acquisition typically relies on attributes or metadata to query for relevant datasets (e.g., filtering rows based on column values), while spatial data acquisition involves spatial relationships and geometry, such as proximity or spatial overlap. The complexity of spatial relationships, combined with the need to consider coverage and connectivity, requires more specialized methods that consider not only attribute similarity but also spatial features and their relative positions in space.

\vspace{-0.1cm}
\subsection{Combinatorial Optimization Problem}
A combinatorial optimization problem involves identifying the optimal combination of variables that maximizes a given objective value from a set of options, while adhering to specific constraints.
In the subsequent sections, we will discuss difference between several classic combinatorial optimization problems that have similar constraints and objectives with \MCPBC, such as the budgeted influence maximization (BIM) problem~\cite{HuangGBL2023,Yuchen2018,Bian2020,Huang2022,Kempe2003,Tang2014Sigmod,HuangKK2019,LiHui2015}, budgeted connected dominating set (BCDS) problem~\cite{Khuller2020,AlberFN04,Guha1998,Wan2002,Gandhi2007,Du2013Connected,ZhangWJ2010}, and maximum coverage problem (MCP)~\cite{Hochbaum1998,Chekuri2004,Rawitz2021,Hochbaum2020,Farbstein2017,Guo2019,Indyk2019}, and illustrate the technical innovations of our proposed method.

\myparagraph{Difference between \MCPBC, BIM, and BCDS} 
The budgeted influence maximization (BIM) problem~\cite{Bian2020,Pierre2020,Banerjee2019,HanZHS2014} aims to select a limited number of seed nodes such that the total cost of nodes is less than the predefined budget, and the influence spread is maximized. While the budgeted connected dominating set (BCDS) problem~\cite{Wan2002,Guha1998,Gandhi2007,Du2013Connected,Khuller2020} is to find $k$ nodes ($k$ is the budget) that induce a connected subgraph, and the number of nodes dominated by the $k$ nodes is maximized. They focus on selecting the combination of sets in order to maximize the objective score, while the definitions of objective scores depend on specific problem contexts. For instance, the objective scores in BIM, BCDS, and \MCPBC correspond to the maximization of influencers, the number of dominated nodes, and the spatial coverage, respectively.

For the BIM, the challenge is how to measure the influence spread of a node since the calculation of influence spread from the seed set was proven to be an NP-complete problem~\cite{Kempe2003}. For the BCDS, the challenge is how to maintain a connected dominating subset such that the number of nodes dominated by the subset is maximized~\cite{Khuller2020}. 
Although these optimization problems appear to have similar constraints, such as budget, connectivity, etc., our problem has a unique challenge: maximizing the union of elements covered by the selected subset under the connectivity and budget constraints. Therefore, the existing approaches to the BIM and BCDS problems cannot address our problem.

\myparagraph{Difference between \MCPBC and MCP} The goal of MCP is to select a limited number of sets, such that the number of the covered elements in these chosen sets is maximal. It is noted that our problem can be viewed as a variant of MCP. Hochbaum and Pathria~\cite{Hochbaum1998} showed that MCP is NP-hard and devised a greedy algorithm with an approximation factor of 1-$1/e$. Moreover, there are many variants that have been extensively studied, such as budgeted MCP (BMCP)~\cite{Khuller1999}, weighted MCP (WMCP)~\cite{Vazirani2001}, connected MCP (CMCP)~\cite{Vandin2011}, and generalized MCP (GMCP)~\cite{CohenK08}. These variants share a common goal, which is to select sets that maximize the total weight of the elements in the union set. However, they differ in the constraints imposed on the solution.

For these combinatorial optimization problems like the MCP and its variants, heuristic algorithms can be employed to find solutions, such as greedy algorithm~\cite{Hochbaum2020}, ant colony optimization~\cite{Dorigo2006}, et al. However, algorithms based on greedy strategies have been shown to have the best possible approximation in polynomial time~\cite{Feige1998}. 
A straightforward adaptation based on the greedy strategy is to first find the dataset with maximum coverage and iteratively choose datasets with the maximum coverage increment that satisfies the spatial connectivity and budget constraints. Since this straightforward approach requires evaluating the spatial connectivity of the resulting dataset at each iteration, this approach can be expensive when the number of datasets is large and lacks any formal theoretical guarantees regarding the approximation ratio.

\noindent\textbf{Novelty of Our Methods.}
The novelty of our methods lies in the proposed greedy and acceleration strategies with theoretical approximation guarantees. Specifically, we first design a \SG algorithm that incorporates an index-based acceleration strategy to handle the evaluation of spatial connectivity more efficiently, thus reducing computational overhead. However, the theoretical analysis shows that the approximation ratio of \textsf{DSA} gradually decreases as the budget increases. To overcome the limitations of \SG, the \BGPA algorithm is introduced based on the designed spatial dataset graph. It iteratively selects paths with common nodes and the maximum marginal gain from the dataset graph. The \BGPA algorithm demonstrates better approximation performance, especially for larger budgets. In addition to the algorithms, we also propose two acceleration strategies in Section~\ref{subsec:accelerate} to enhance the efficiency of these algorithms.

\end{sloppypar}

\section{Problem Formulation}
\label{sec:defnitions}

\begin{sloppypar}

In this section, we introduce the data modeling and formalize the problem. A summary of frequently used notations is provided in Table~\ref{tab:notations}.

\begin{figure*}
\setlength{\abovecaptionskip}{0 cm}
\setlength{\belowcaptionskip}{-0.3 cm}	
\centering
\includegraphics[width=0.7\textwidth]{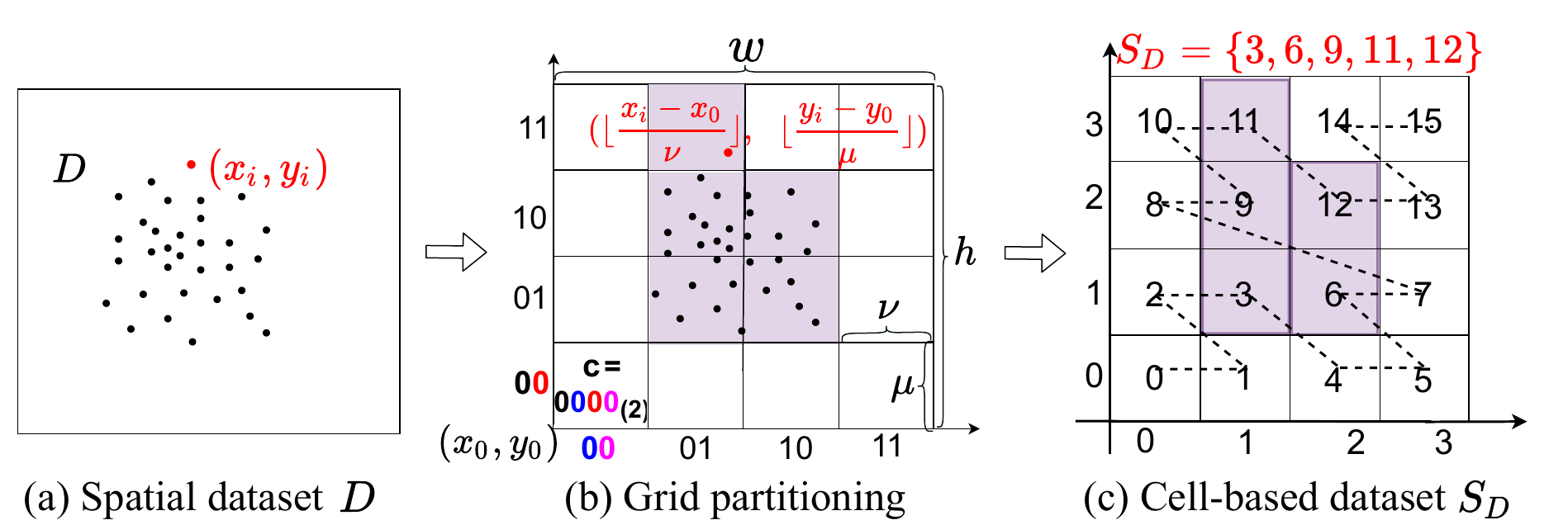}
\caption{Illustration of the spatial dataset and  \SpatialSet, where (a) shows a spatial dataset $D$, (b) shows how to partition the original space into a grid of uniform cells, and (c) shows a cell-based dataset $S_D$ and its spatial coverage.}
\label{fig:rasterPreprocessing}
\end{figure*}

\subsection{Data Modeling}

\begin{definition1}
\textbf{(Spatial Dataset)} A spatial dataset $D$ contains a sequence of points or locations represented by two-dimensional coordinates, i.e., $D = \{(x_1, y_1), (x_2, y_2), \dots, (x_{|D|}, y_{|D|})\}$.
\end{definition1}

\begin{definition1}
\textbf{(Spatial Dataset Marketplace)} A spatial dataset marketplace $\mathcal{D}$ contains a set of spatial datasets, i.e., $\mathcal{D} = \{D_1, D_2, \dots, D_{|\mathcal{D}|}\}$. 
\end{definition1}

To quantify the coverage of a spatial dataset, we partition the space into a grid of uniform cells~\cite{Zacharatou2017,Cao2021a}, since the individual points have no length or width. This approach provides a more effective means of measuring dataset coverage.
Specifically, given an integer $\theta$, we partition the entire space containing all spatial datasets into a grid $\mathcal{C}_{\theta}$ of $2^\theta\times2^\theta$ uniform cells, where $\theta$ is called the \textit{resolution}.

Each cell, denoted as $c$, denotes a unit space in the grid $\mathcal{C}_{\theta}$. Its coordinates $(x, y)$ can be converted into a non-negative integer to uniquely indicate $c$ using the z-order curve method~\cite{Peng2016, Yang2022}, denoted as $z(x, y) = c$. This transformation yields consecutive IDs in the range $[0,2^{\theta} \times 2^{\theta}-1]$. In this way, a spatial dataset $D$ can be represented as a sorted integer set consisting of a sequence of cell IDs, referred to as a cell-based dataset. The spatial coverage of a dataset is the number of cell IDs in its corresponding cell-based dataset. The formal definition of \textit{cell-based dataset} and \textit{spatial coverage} is as follows.

\begin{definition1}
\textbf{(Cell-based Dataset and Spatial Coverage)} 
Given a 2-dimensional space and a grid $\mathcal{C}_{\theta}$, a cell-based dataset $S_{D,\mathcal{C}_{\theta}}$ of the dataset $D$ is a set of cell IDs where each cell ID denotes that there exists at least one point $(x, y)\in D$ falling into the cell $c$ of the grid $\mathcal{C}_{\theta}$. The coordinates $(x, y)$ of cell $c$ are $(\frac{x-x_0}{\nu}, \frac{y-y_0}{\mu})$, where $(x_0,y_0)$ represents the coordinate of the bottom-left point in the whole 2-dimensional space, and $\nu$ and $\mu$ denote the width and height of the cell $c$. Thus, the cell-based dataset is represented as $S_{D,\mathcal{C}_{\theta}} = \{ z(\frac{x-x_0}{\nu}, \frac{y-y_0}{\mu}) | (x,y) \in D \}$. The spatial coverage of $S_{D,\mathcal{C}_{\theta}}$ is represented by the number of cells, denoted as $|S_{D,\mathcal{C}_{\theta}}|$.   For ease of presentation, we omit the subscript $\mathcal{C}_{\theta}$ of $S$ when the context is clear.
\end{definition1}

\begin{example1}
Fig.~\ref{fig:rasterPreprocessing}(a) first shows an example of a raw spatial dataset $D$ in a rectangular space. Each point of $D$ can be mapped to the cell of the grid, as shown in Fig.~\ref{fig:rasterPreprocessing}(b). Each cell can be uniquely represented by an integer, e.g., for the bottom left cell, its ID is $0$, which is transformed from its coordinates (0, 0). Thus, the spatial dataset $D$ can be represented as a finite set $S_D$ consisting of a sequence of cell IDs, i.e., $S_D =\{3, 6, 9, 11, 12\}$, as shown in Fig.~\ref{fig:rasterPreprocessing}(c). The spatial coverage of $S_D$ is the number of IDs in the set, which is 5.
\end{example1}

\begin{table}
\renewcommand{\arraystretch}{1.0}
\setlength{\abovecaptionskip}{0.5cm} \caption{Summary of notations.}
\vspace{-0.3cm}
 \label{tab:notations}
\scalebox{0.95}{ \begin{tabular}{cc}  
\toprule   
  \textbf{Notation} & \textbf{Description}  \\  
\midrule   
$D$ & a spatial dataset\\
$\mathcal{D}$ & a spatial data marketplace\\
$|\mathcal{D}|$ & the number of datasets in $\mathcal{D}$ \\
$S_D$ & the \SpatialSet of $D$\\  
$p(S_D)$ & the price of the \SpatialSet $S_D$ \\
$\mathcal{S}_{\mathcal{D}}$ & the collection of \SpatialSets \\
$\mathcal{G}(\mathcal{V},\mathcal{E})$ & the dataset graph with vertices $v$ and edges $e$\\
B & the budget \\
$\delta$ & the distance threshold \\
$dist(S_{D_i}, S_{D_j})$ & the distance between $S_{D_i}$ and $S_{D_j}$ \\
  \bottomrule  
\end{tabular}}
\end{table}

In order to quantify the relationship between two spatial datasets, it is necessary to employ a distance measure for spatial datasets. Several distance functions~\cite{SuLZZZ2020,Nutanong2011,Yang2022} have been proposed for this purpose, such as the Hausdorff~\cite{Nutanong2011} and the EMD distance~\cite{Yang2022}. Note that we can adopt any suitable distance function to measure the distance between datasets based on applicable scenarios. 

In this paper, we define a cell-based dataset distance metric by extending the $MinMinDist$ \cite{Corral2000}. The goal is to capture the proximity between datasets in space by computing the nearest neighbor distance between the cells of two cell-based datasets. For example, a cell-based dataset distance of zero indicates that the two cell-based datasets intersect.

\begin{definition1}
\label{defi:dist}
\textbf{(Cell-based Dataset Distance)} Given two cell-based datasets $S_{D_i}$ and $S_{D_j}$, their distance $dist(S_{D_i}, S_{D_j})$ is the minimum distance between their cells and is defined as:
\begin{equation}
\small
      dist(S_{D_i}, S_{D_j}) = \min_{c \in S_{D_i}, c' \in S_{D_j}} \: ||c, c'||_2
\end{equation}
where $||c, c'||_2$ denotes the Euclidean distance between the coordinates represented by $c'$ and $c$. 
\end{definition1}

In addition, to evaluate whether a set of spatial datasets satisfies the spatial connectivity, it is essential to establish criteria defining the relationships between pairs of datasets. Thus, we introduce the concepts for defining relationships between datasets. Thus, we introduce the concept of spatial connectivity, which encompasses both direct and indirect connections.

\begin{definition1}
\label{defi:connectivity}
\textbf{(Spatial Connectivity)} Given a collection of cell-based datasets $\mathcal{S}_{\mathcal{D}}=\{S_{D_1}, \ldots, S_{D_{|\mathcal{D}|}}\}$ constructed from $\mathcal{D}$ and a grid $\mathcal{C}_{\theta}$, the collection $\mathcal{S}_{\mathcal{D}}$ satisfies the spatial connectivity if, for a given distance threshold $\delta$, any pair of datasets $S_{D_i}\in \mathcal{S}_{\mathcal{D}}$ and $S_{D_j}\in \mathcal{S}_{\mathcal{D}}$ $(i \neq j)$ are either directly connected (i.e., $dist(S_{D_i}, S_{D_j}) \leq \boldsymbol{\delta}$) or indirectly connected through a sequence of ordered datasets $\{S_{D_m}, \dots, S_{D_n}\}$, where each pair of adjacent datasets (i.e., ($S_{D_i}, S_{D_m}$), \dots, ($S_{D_n}, S_{D_j}$))  directly connected.
\end{definition1}

\begin{figure*}
\setlength{\abovecaptionskip}{0 cm}
\setlength{\belowcaptionskip}{-0.3 cm}	
\centering
\includegraphics[width=0.75\textwidth]{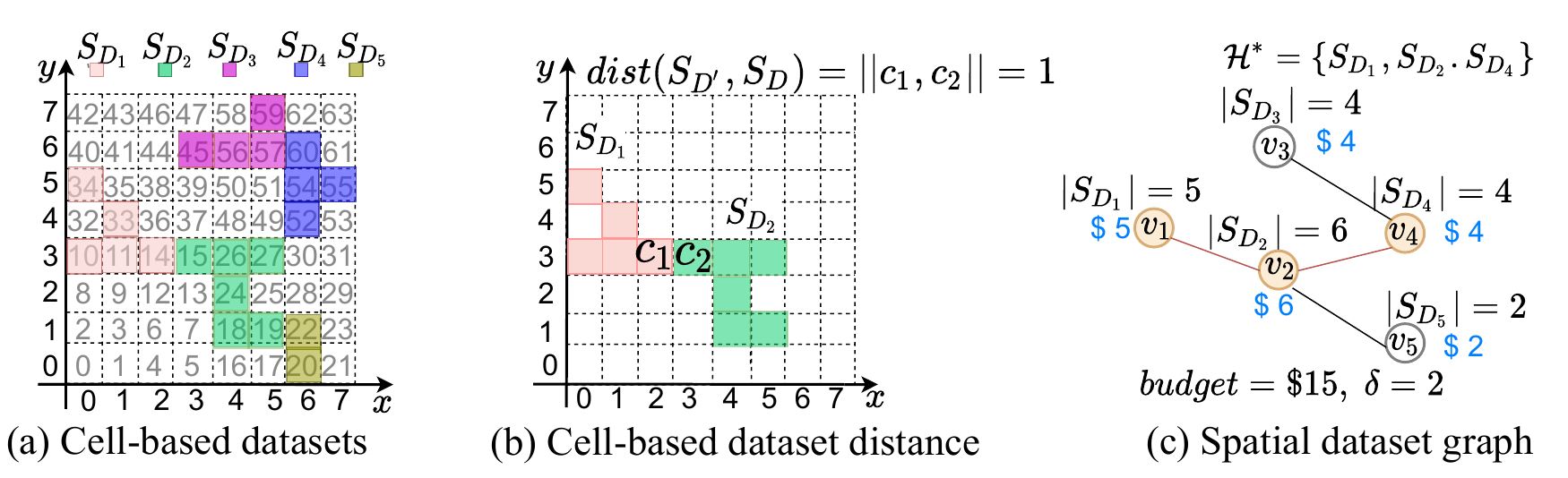}
\caption{
Construction of spatial dataset graph, where (a) shows five cell-based datasets in $\market$, (b) shows the cell-based dataset distance computation between $S_{D_1}$ and $S_{D_2}$, and (c) shows the spatial dataset graph constructed from $\mathcal{S}_{\mathcal{D}}$ and the optimal solution of $\mathcal{H}^*$.
}
\label{fig:connectedGraph}
\end{figure*}

\begin{definition1}
\label{def:datasetGraph}
\textbf{((Connected) Spatial Dataset Graph)}
Given a set of cell-based datasets $\mathcal{S}_{\mathcal{D}}=\{S_{D_1},\ldots,S_{D_{|\mathcal{D}|}}\}$ constructed from $\mathcal{D}$ and a grid $\mathcal{C}_{\theta}$, a pricing function $p$, and a distance threshold $\delta$, we can construct an undirected graph 
$\mathcal{G}_{\mathcal{S}_{\mathcal{D}},\delta}=(\mathcal{V},\mathcal{E})$, where each $v_i \in \mathcal{V}$ denotes a \SpatialSet $S_{D_i} \in \mathcal{S}_{\mathcal{D}}$ with a fixed price $p(S_{D_i})$ and there is an edge between $v_i$ and $v_j \in \mathcal{V}$ $(i\neq j)$, if and only if $S_{D_i}$ and $S_{D_j}$ are directly connected relation. 
The graph $\mathcal{G}_{\mathcal{S}_{\mathcal{D}},\delta}$ is a \textbf{connected} spatial dataset graph if the collection $\mathcal{S}_{\mathcal{D}}$ satisfies the spatial connectivity.

\end{definition1}

\subsection{Problem Definition}
\begin{definition1}
\label{def:MCPBC}
\textbf{(Budgeted Maximum Coverage with Connectivity Constraint (\MCPBC))}
Given a spatial dataset graph $\mathcal{G}_{\mathcal{S}_{\mathcal{D}},\delta}=(\mathcal{V},\mathcal{E})$, 
a distance threshold $\delta$, a pricing function $p$ and a budget $B$, we aim to find an optimal subset $\mathcal{H}^*\subseteq \mathcal{V}$ such that the number of cells covered by $\mathcal{H}^*$ is maximized:
\begin{equation}
\small
    \mathcal{H^*}= \argmax_{\mathcal{H} \subseteq \mathcal{V}} |\bigcup_{v_i \in \mathcal{H}} S_{D_i}|
\end{equation}

\noindent subject to:
\begin{enumerate}
    \item the cost of datasets in $\mathcal{H}^*$ is under $B$: $\sum\limits_{v_i\in \mathcal{H}^*} p(S_{D_i}) \leq B$;
    \item the dataset graph $\mathcal{G}_{\mathcal{H^*},\delta}$ constructed from $\mathcal{H}^*$ is a connected spatial dataset graph;
\end{enumerate}
\end{definition1}

\noindent where each cell-based dataset $S_D$ is associated with a price $p(S_{D})$. The price can be set by the data seller or determined by any pricing functions~\cite{Pei2022}. 
For ease of understanding, we adopt a usage-based pricing~\cite{Pei2022,Miao2022,AzcoitiaIL2023} function in the next example, where the price of the spatial dataset is set to the number of covered cells.

\begin{example1}
  Suppose we have a marketplace $\mathcal{D} = \{D_1, \dots, D_5\}$, the corresponding dataset prices are $\{\$5, \$6, \$4, \$4, \$2\}$.
  Fig.~\ref{fig:connectedGraph}(a) shows the corresponding cell-based datasets $\mathcal{S}_{\mathcal{D}}$ constructed form $\mathcal{D}$ and a grid $\mathcal{C}_{3}$. We then construct the spatial dataset graph $\mathcal{D}_{\mathcal{S_D}, \delta}$ with $\delta =2$. For one of the dataset pairs $S_{D_1}$ and $S_{D_2}$, we can obtain the cell-based dataset distance by calculating the distance between their nearest neighbor cells, as shown in Fig.~\ref{fig:connectedGraph}(b). Since the distance $dist(S_{D_1}, S_{D_2}) = 1 < \delta$, $S_{D_1}$ and $S_{D_2}$ are directly connected. After computing the distance between all pairs of datasets, we can obtain the spatial dataset graph, as shown in  Fig.~\ref{fig:connectedGraph}(c).
  A data buyer with a budget of \$15 wants to buy a set of connected datasets. From all feasible solutions of \MCPBC we can obtain the optimal subset $\mathcal{H}^* = \{S_{D_1}, S_{D_2}, S_{D_4}\}$.

\end{example1}
\end{sloppypar}

\section{NP-hardness of \MCPBC}
\label{sec:nphard}

\begin{figure*}
\setlength{\abovecaptionskip}{0.2 cm}
\setlength{\belowcaptionskip}{-0.3 cm}	
\centering
\includegraphics[width=0.9\textwidth]{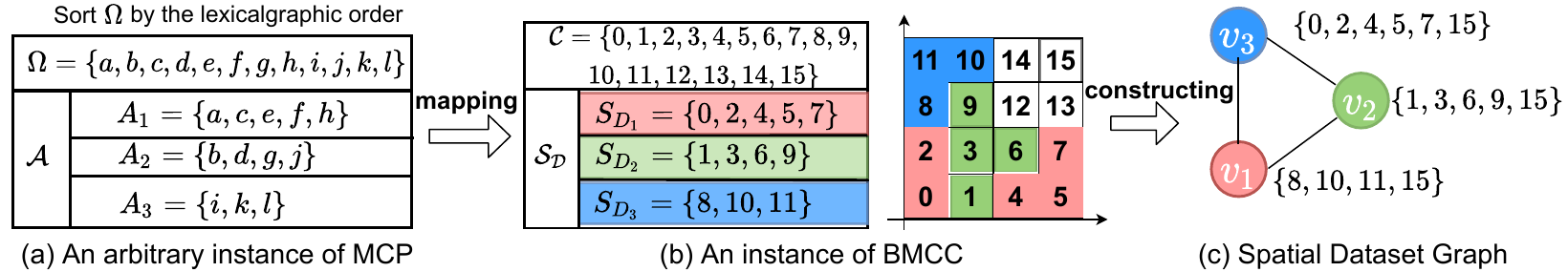}
\caption{Illustration of reduction from MCP to \MCPBC.
}
\label{fig:nphard}
\vspace{-0.3cm}
\end{figure*}

\begin{theorem}
\label{theorem:nphard}
    The \textbf{BMCC} problem is NP-hard.
\end{theorem}

\begin{proof}
To prove the NP-hardness of the \MCPBC, we prove that any instance of Maximum Coverage Problem~\cite{Hochbaum1998} (MCP) can be reduced to an instance of the \MCPBC. 
Given a sorted universe of elements $\Omega = \{o_1, o_2, \dots, o_{|\Omega|}\}$, a collection of sets $\mathcal{A} = \{A_1, A_2, \dots, A_{|\mathcal{A}|}\}$ ($\forall A\in \mathcal{A}, A \subseteq \Omega$) and a positive integer $k$, the objective of the MCP is to find at most $k$ sets $\{A_1', A_2', \dots, A_k'\} \subseteq \mathcal{A}$ such that the $|\bigcup_{i=1}^k A_i'|$ is maximized. We show that the MCP can be viewed as a special case of the \MCPBC.

Given an arbitrary instance of the MCP, we sort the universe $\Omega$ according to the lexicographic order, then create a mapping from the sorted universe $\Omega = \{o_1, o_2, \dots, o_{|\Omega|}\}$ to an integer set $ \{0, 1, \dots, |\Omega|-1\}$. Then we can construct a $2^\theta \times 2^\theta$ grid in the space such that $2^\theta \times 2^\theta \geq |\Omega|$.
Since each cell in the grid can be represented by a unique integer ID, each element in $\Omega$ can be mapped to a cell whose ID is equal to the integer, as shown in Figs.~\ref{fig:nphard}(a) and \ref{fig:nphard}(b). The other four cells with IDs 12, 13, 14, and 15 are not included in $\forall S_{D_i} \in \mathcal{S}_\mathcal{D}$, indicating that no datasets are located in those cells.

Then we set the price of each vertex to 1, the budget to $k$, and $\delta$ to $2^{\theta}\sqrt{2}$ in Definition~\ref{def:MCPBC}. 
The set $\mathcal{S}_{\mathcal{D}}$ can be modeled into a complete undirected graph $\mathcal{G}_{\mathcal{S}_{\mathcal{D}, \delta}} = (\mathcal{V}, \mathcal{E})$: each vertex $v_i$ corresponding to each \SpatialSet $S_{D_i} \in \market$; for each vertex $v_i$ in the graph, there is an undirected edge $e_{i, j}$ $(\forall v_j \in \mathcal{V}$, $i \neq j)$. This is because $dist(S_{D_i}, S_{D_j})$ must not exceed $\delta$, ensuring the connectivity constraint is always satisfied. Fig.~\ref{fig:nphard}(c) shows a spatial dataset graph constructed from the MCP instance, and we can see that the total reduction is performed in polynomial time.

The $\MCPBC$ is to find at most $k$ vertices in this complete graph $\mathcal{G}_{\mathcal{S}_{\mathcal{D}, \delta}}$ such that the size of the union of the elements contained in $k$ vertices is maximized, which is equivalent to finding at most $k$ sets from $\mathcal{S}$ such that the $|\bigcup_{i=1}^k S_i'|$ is maximized in MCP. The MCP has been proven to be NP-hard~\cite{Hochbaum1998}.
Therefore, if we can find the optimal solution for the $\MCPBC$ instance in polynomial time, the MCP can be solved in polynomial time, which is not possible unless P=NP. Therefore, the $\MCPBC$ is NP-hard.
\end{proof}

\section{Approximation Algorithms for \MCPBC}
\label{sec:algorithms}

\begin{sloppypar}

In this section, we focus on solving the \MCPBC problem (defined in Def.~\ref{def:MCPBC}). By Theorem~\ref{theorem:nphard}, we know that this problem is NP-hard. Achieving the optimal solution for \MCPBC is computationally prohibitive. Therefore, we first propose a greedy algorithm, which provides a $\frac{p_{min}}{B}$ approximation guarantee (see Section~\ref{subsec:SG}). However, its approximation ratio decreases when confronted with a large budget $B$. To address this, we propose an improved greedy algorithm that produces a superior approximation ratio for scenarios involving large $B$ (see Section~\ref{sub:IBGPA}).

\subsection{Dual-Search Algorithm (\SG)}
\label{subsec:SG}
In this subsection, we present a \emph{Dual-Search Algorithm} called \SG. The basic idea of \SG is to perform a two-round heuristic search based on the greedy strategy, the reason for performing the two-round search is primarily the need to balance the three factors of spatial coverage, budget, and spatial connectivity. The goal of the first-round search is to make optimal use of the budget, that is, to maximize the increment of spatial coverage per unit budget. Specifically, the algorithm selects the dataset that maximizes the price-to-coverage ratio (e.g., the ratio of the spatial coverage increment over the price of the dataset) while maintaining the spatial connectivity at each step.

The first-round search ensures efficient use of the budget by considering the price-to-coverage ratio. However, due to the spatial connectivity constraint, it may overlook datasets with larger spatial coverage. For instance, if all datasets have the same price-to-coverage ratio, the algorithm cannot effectively differentiate their cost performance, potentially missing datasets with larger coverage. This results in a locally optimal solution. Therefore, the second-round search iteratively chooses the set that maximizes the marginal gain (e.g., the spatial coverage increment) while maintaining the spatial connectivity.

The second-round search complements the dataset ignored in the first round by maximizing coverage, ensuring maximum spatial coverage of the final result. The optimal selections in the two-round search can result in two feasible solutions, which jointly guarantee the effective approximation ratio of \SG whenever the dataset chooses any pricing function, avoiding local optimal solutions. In the following, we will introduce how \SG works according to the pseudocode of Algorithm~\ref{alg:simplegreedy} and then theoretically prove the approximation ratio of \SG. Afterwards, we will further analyze its time complexity in detail.

\begin{figure*}
\setlength{\abovecaptionskip}{-0.2 cm}
\setlength{\belowcaptionskip}{-0.6 cm}	
\centering
\includegraphics[width=0.65\textwidth]{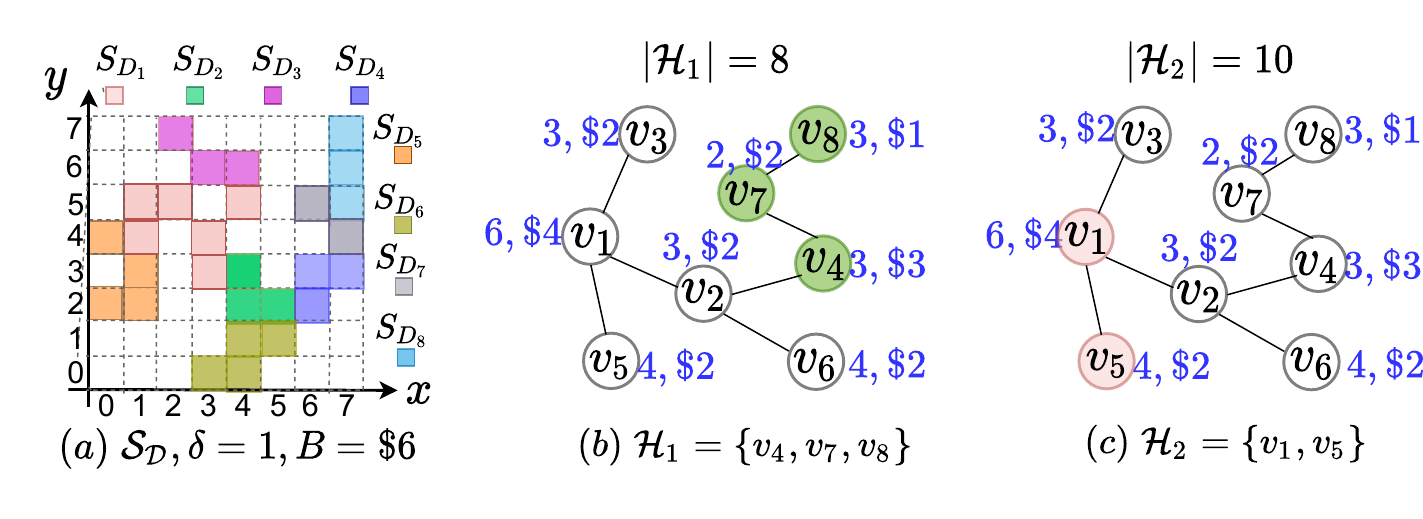}
\caption{
Illustration of \SG, where (a) shows 8 cell-based datasets, (b) shows the solution $\mathcal{H}_1=\{S_{D_4}, S_{D_7}, S_{D_8}\}$ found by the first-round search, and (c) shows the solution $\mathcal{H}_2 = \{S_{D_1},  S_{D_5}\}$ found by the second-round search.
}
\label{fig:simpleGreedy}
\end{figure*}

\begin{algorithm}[t]
\small
\caption{$\texttt{\SG}(\market, p, B, \delta)$}
\label{alg:simplegreedy}
\LinesNumbered 
\KwIn{$\market$: a set of \SpatialSets, $p$: a price function for each $S_{D_i}\in \market$,
  $B:$ a budget,
  $\delta$: a distance threshold. }
  \KwOut{$\mathcal{H}$: the result subset.}
$\mathcal{U}_1\leftarrow \cup_{S_{D_i} \in \market} S_{D_i}$; $\mathcal{U}_2 \leftarrow \mathcal{U}_1$\;
$\market' \leftarrow \emptyset; \mathcal{H}_1  \leftarrow \emptyset; \mathcal{H}_2  \leftarrow \emptyset
$\;

\ForEach{$S_{D_i} \in \market$}{ \label{line:startFliter1}
\If{$p(S_{D_i}) \leq B$}{
$\market' \leftarrow \market'\cup \{S_{D_i}\}$\;
}
}\label{line:startFliter2}
$\market'' \leftarrow \market'$;\Comment{make a copy of $\market$}\\
\While{$\market' \neq \emptyset$ and $\sum\limits_{S\in \mathcal{H}_1}p(S)\leq B$}{
 Select a set $S_{D_i} \in \market'$ that maximize $\frac{|S_{D_i} \cap \mathcal{U}_1|}{p(S_{D_i})}$\; \label{line:maxCoverage1}
\If{$\mathcal{G}_{\mathcal{H}_1 \cup \{S_{D_i}\}, \delta}$ is connected 
        }{
        \If{$\sum\limits_{S\in \mathcal{H}_1}p(S) + p(S_{D_i}) \leq B$}{
       
        $\mathcal{H}_1 = \mathcal{H}_1 \cup \{S_{D_i}\}$\;  
        $\mathcal{U}_1 = \mathcal{U}_1\backslash S_{D_i} $\;
        }
        }
        $\market' = \market'\backslash S_{D_i} $\;\label{line:maxCoverage2}
        }
\While{$\market'' \neq \emptyset$ and $\sum\limits_{S\in \mathcal{H}_2}p(S)\leq B$}{
Select a set $S_{D_i} \in \market''$ that maximize $|S_{D_i} \cap \mathcal{U}_2|$\; \label{line:maxCoverage11}
\If{$\mathcal{G}_{\mathcal{H}_2 \cup \{S_{D_i}\}, \delta}$ is connected 
        }{
        \If{$\sum\limits_{S\in \mathcal{H}_2}p(S) + p(S_{D_i}) \leq B$}{
       
        $\mathcal{H}_2 = \mathcal{H}_2 \cup \{S_{D_i}\}$\;  
        $\mathcal{U}_2 = \mathcal{U}_2 \backslash S_{D_i} $\;
        }
        }
        $\market'' = \market'' \backslash S_{D_i} $\;\label{line:maxCoverage22}
        } 

$\mathcal{H} \leftarrow \argmax\limits_{\mathcal{H}\in \{\mathcal{H}_1, \mathcal{H}_2\}} |\bigcup\limits_{S\in \mathcal{H}}S|$ \;\label{line:pickmax1}
\Return $\mathcal{H}$\; \label{line:pickmax2}
\end{algorithm}

\noindent\textbf{First-round Search.} We start by filtering out datasets in the data marketplace $\market$ that are out of the budget and obtain a set of candidate datasets $\market'$ (see Lines~\ref{line:startFliter1} to \ref{line:startFliter2}). We then perform a two-round search on $\market'$. Specifically, during the first-round search (see Lines~\ref{line:maxCoverage1} to \ref{line:maxCoverage2}), the algorithm iteratively selects a cell-based dataset $S_{D_i}\in \market'$ to add to $\mathcal{H}_1$, which not only satisfies the budget and spatial connectivity constraints, but also has the maximum ratio of the spatial coverage increment to the dataset's price $|S_{D_i} \cap \mathcal{U}_1|/{p(S_{D_i})}$, until the budget is used up.

\noindent\textbf{Second-round Search.}
It is worth noting that choosing a good starting dataset from $\mathcal{S}_{\mathcal{D}}'$ is crucial due to the spatial connectivity constraint. For instance, if the first selected node is an isolated node (i.e., there is no node connected to it), then it will lead to poor results. To avoid this situation, we increase the number of feasible solutions to ensure the algorithm's effectiveness and prevent the algorithm from falling into local optimality.
Algorithm~\ref{alg:simplegreedy} uses the $\market''$ obtained by copying $\market'$ to perform the second-round search (see Lines~\ref{line:maxCoverage11} to \ref{line:maxCoverage22}). It iteratively selects the dataset with the maximum coverage increment $|S_{D_i} \cap \mathcal{U}_2|$ while ensuring spatial connectivity, until the budget is exhausted.
Finally, the algorithm returns the one with the maximum spatial coverage from $\mathcal{H}_1$ and $\mathcal{H}_2$ (see Lines~\ref{line:pickmax1} to \ref{line:pickmax2}).

\begin{example1}
    Fig.~\ref{fig:simpleGreedy}(a) shows a set of cell-based datasets $\mathcal{S}_{\mathcal{D}} = \{S_{D_1}, \dots, S_{D_8}\}$, a distance threshold $\delta = 1$ and a budget $B = \$6$. 
    When performing \SG, we first find the candidates $\market'$. Since the price of each dataset $S_D \in \market$ is less than $B$, the candidate datasets $\market'= \market$ in this example. The first-round search preferentially selects the dataset with the maximum ratio $|S_{D_i} \cap \mathcal{U}_1|/{p(S_{D_i})}$ within the budget in each iteration. As shown in Fig.~\ref{fig:simpleGreedy}(b), the algorithm adds $\{v_8, v_7, v_4\}$ to the result subset $\mathcal{H}_1$ one by one until the budget is exhausted, where the spatial coverage of $\mathcal{H}_1$ is 8. Next, the second-round search iteratively selects the dataset with the maximum coverage increment $|S_{D_i} \cap \mathcal{U}|$ and satisfies spatial connectivity. Thus, we add the set $\{v_1, v_5\}$ into $\mathcal{H}_2$, resulting in a spatial coverage of 10 for $\mathcal{H}_2$, as shown in Fig.~\ref{fig:simpleGreedy}(c). Finally, the algorithm returns the result $\mathcal{H}_2$, which contains a larger spatial coverage than $\mathcal{H}_1$.
\end{example1}

\noindent\textbf{Theoretical Analysis.} In the following, we analyze the approximation ratio of Algorithm~\ref{alg:simplegreedy}. We use $p_{max}$ and $p_{min}$ to represent the maximum and minimum prices of the datasets in $\market$, respectively.
Without loss of generality, we can assume that the budget is greater than the minimum price (i.e., $B \geq p_{min}$).
To find a result subset that satisfies connectivity, the solution for each round of search contains at least one dataset. We have the following theorem.

\begin{theorem}
\label{theorem:simpleGreedy}
     \SG delivers a $\frac{p_{min}}{B}$ approximate solution for the \MCPBC.
\end{theorem}

\begin{proof}

Let $\mathcal{H}^* = \{S_{D_1}^*, \dots, S_{D_{|\mathcal{H}^*|}}^*\}$ be the optimal solution of \MCPBC, $|OPT| = |\bigcup\limits_{S \in \mathcal{H}^*} S|$.

(1) We first analyze the approximation ratio of the first-round search of Algorithm~\ref{alg:simplegreedy}. At the end of iteration $1$, the algorithm first selects the cell-based dataset $S_{D_1}' \in \market'$ that satisfies all constraints and has the maximum ratio of spatial coverage to the dataset price $\frac{|S_{D_1}'|}{p(S_{D_1}')}$. Thus we have
\begin{small}
\begin{equation}
\label{eq:oSetPrice}
\small
    \forall S_{D_j}^* \in \mathcal{H}^*, \frac{|S_{D_j}^*|}{p(S_{D_j}^*)} \leq \frac{|S_{D_1}'|}{p(S_{D_1}')}.
\end{equation}
\end{small}

The Inequality~\ref{eq:oSetPrice} is equivalent to 
\begin{small}
\begin{equation}
\label{eq:oSetPrice2}
    \forall S_{D_j}^* \in \mathcal{H}^*, |S_{D_j}^*| \times \frac{p(S_{D_1}')}{|S_{D_1}'|} \leq {p(S_{D_j}^*)}.
\end{equation}
\end{small}

Since the total price of the sets in $\mathcal{H}^*$ is bounded by the budget $B$, which implies
\begin{equation}
\label{eq:lessbudget}
\small
    p(S_{D_1}^*) + \dots + p(S_{D_{|\mathcal{H}^*|}}^*) \leq B.
\end{equation}

By Inequalities~\ref{eq:oSetPrice2} and \ref{eq:lessbudget}, we obtain that
    \begin{equation}
    \small
      (|S_{D_1}^*| + \dots + |S_{D_{|\mathcal{H}^*|}}^*|) \times \frac{p(S_{D_1}')}{|S_{D_1}'|} \leq B.
\end{equation}

Since $|OPT| \leq (|S_{D_1}^*| + \dots + |S_{D_{|\mathcal{H}^*|}}^*|)$, we can obtain that
  \begin{equation}
  \small
|OPT| \times \frac{p(S_{D_1}')}{|S_{D_1}'|} \leq B.
\end{equation}

Multiplying both sides by $\frac{|S_{D_1}'|}{B}$, we get the inequality
  \begin{equation}
  \small
     |S_{D_1}'| \geq \frac{p(S_{D_1}')}{B} \times |OPT| \geq \frac{p_{min}}{B} \times |OPT|.
\end{equation}

(2) Next, we analyze the approximation ratio of the second-round search of Algorithm~\ref{alg:simplegreedy}. We assume the maximum size of the result subset found by the second-round search is $k = 
 \frac{B}{p_{min}}$. 
 For the optimal solution $\mathcal{H}^*$, we have 
\begin{small}
     \begin{equation}
    \begin{aligned}
    \label{eq:opt_v}
        \sum_{S \in \mathcal{H}^*} |S| & \geq |\bigcup_{S \in \mathcal{H}^*} S| = |OPT|. \\
    \end{aligned}
    \end{equation}
\end{small}
   
Dividing the two sides of Inequality~\ref{eq:opt_v} by $k$, we obtain that
\begin{small}
    \begin{equation}
    \begin{aligned}
        \frac{\sum_{S \in \mathcal{H}^*} |S|}{k} 
        &\geq \frac{|OPT|}{k}. \\
    \end{aligned}
    \end{equation}
\end{small}

 Due to the maximum spatial coverage of a set in the result subset not being less than the average spatial coverage of all sets, we have
 \begin{small}
  \begin{equation}
    \begin{aligned}
        \max_{S\in \mathcal{H}} |S| \geq \frac{\sum_{S \in \mathcal{H}^*} |S|}{k} \geq \frac{|OPT|}{k}.
    \end{aligned}
\end{equation}   
 \end{small}

Because at the end of iteration $1$, the second-round search first selects the dataset $S_{D_1}''$ with the maximum spatial coverage in $\market''$. We have
 \begin{small}
     \begin{equation}
    \label{eq:simpleConclusion2}
        |S_{D_1}''| \geq \frac{|OPT|}{k} \geq \frac{p_{min}}{B} |OPT|.
    \end{equation}
 \end{small}
The above analysis proves that \SG provides a $\frac{p_{min}}{B}-$ approximation ratio for the \MCPBC.
\end{proof}

\vspace{-0.3cm}
\noindent\textbf{Time Complexity.} We denote $n = |\market|$ as the number of datasets in the $\market$.
Firstly, identifying a set of candidate datasets $\market'$ whose price is lower than the budget can be accomplished in $O(n)$. 
In the first-round search, there are at most $\frac{B}{p_{min}}$ iterations, with each iteration requiring $O(n)$ set union operations to pick the subset that has the maximum ratio $|S_{D}' \cap \mathcal{U}|/p(S_{D}')$. Here, we assume that each cell-based dataset $S_D' \in \market'$ is represented by an array of ordered indexes of length $|S_D'|$. Consequently, obtaining the union set and measuring its size for two datasets requires $O(|\mathcal{U}|)$ time, where $|\mathcal{U}| = |\cup_{S_D \in \market'} S_D|$. Specifically, given a universe $\mathcal{U}$ and two cell-based datasets $S_{D_1}$ and $S_{D_2}$, we can traverse every cell ID of the $\mathcal{U}$ and determine whether it exists in either $S_{D_1}$ or $S_{D_2}$. Thus, the time complexity of finding the union set is $O(|\mathcal{U}|)$. The second-round search has the same time complexity as the first-round search because it follows the same process, except that the dataset with the maximum coverage increment $|S_{D}' \cap \mathcal{U}|$ is selected at each iteration. The overall complexity is $O(n + 2\frac{B}{p_{min}}|\mathcal{U}|n) = O(\frac{B|\mathcal{U}|}{p_{min}}n)$.

 \vspace{-0.2cm}
\subsection{Improved Algorithm of \SG }
\label{sub:IBGPA}

\begin{algorithm}[h!]
\small
\caption{$\texttt{\BGPA}(\market, p, B, \delta)$}
\label{alg:IBGPA}
\LinesNumbered 
\KwIn{$\market$: a set of \SpatialSets,
$p$: a price function for each $S_{D_i}\in \market$,
  $B:$ a budget,
  $\delta$: a distance threshold. }
\KwOut{$\mathcal{H}$: the result subset.
}

$\mathcal{H} \leftarrow \emptyset$;
$\market' \leftarrow \emptyset$\;

\ForEach{$S_{D_i} \in \market$}{ \label{line:DPSAFliter1}
\If{$p(S_{D_i}) \leq B $}{
$\market' \leftarrow \market'\cup \{S_{D_i}\}$\;
}
}\label{line:DPSAFliter2}
$\mathcal{G}_{\market', \delta} \leftarrow $ build the spatial dataset graph over $\market'$ and  $\delta$ and use BFS to find all subgraphs\;
\label{line:constructGraph}

\ForEach{$\mathcal{G}_{\market', \delta}^i \in \mathcal{G}_{\market', \delta}$}{ \label{line:eachSubgraph}
Run BFS rooted at each node in $\mathcal{G}_{\market', \delta}$ to get the center node $v^*$ and the BFS tree $T_{v^*}$\; \label{line:findCenter}
$\mathcal{L}_{T_{v^*}} \leftarrow$ set of leaves of the BFS tree $T_{v^*}$\;
$\mathcal{S}_{\mathcal{L}_{T_{v^*}}} \leftarrow$ set of $S_{L_{(v^*, v_j)}}$ for each  $v_j \in \mathcal{L}_{T_{v^*}}$\;
\scalebox{0.9}{$\mathcal{H}_1 \leftarrow \texttt{BudgetedGreedy}(\market', \mathcal{L}_{T_{v^*}}, \mathcal{S}_{\mathcal{L}_{T_{v^*}}}, v^*, B, 0)$}\;
\label{line:addWhile1}
\scalebox{0.9}{$ \mathcal{H}_2 \leftarrow \texttt{BudgetedGreedy}(\market', \mathcal{L}_{T_{v^*}}, \mathcal{S}_{\mathcal{L}_{T_{v^*}}}, v^*, B, 1)$}\; \label{line:addWhile}
}
$\mathcal{H} \leftarrow \argmax\limits_{\mathcal{H}\in \{\mathcal{H}_1, \mathcal{H}_2\}} |\bigcup\limits_{S\in \mathcal{H}}S|$ \;

\Return $\mathcal{H}$\;
\vspace{-0.6em}
\noindent\rule{8cm}{0.6pt}
  \SetKwFunction{FSum}{\scalebox{0.9}{BudgetedGreedy}}\label{ExhaustBudgetedGreedy}
 \SetKwProg{Fn}{Function}{:}{}
  \Fn{\FSum{\scalebox{0.80}{$\market',\mathcal{L}_{T_{v^*}}, \mathcal{S}_{\mathcal{L}_{T_{v^*}}},v^*, $  $B,$ \hspace{-0.2cm} $ flag$}}}{
  \KwIn{
  $\market'$: a set of candidate datasets,
  $\mathcal{L}_{T_{v^*}}$: set of leaves of the BFS tree $T_{v^*}$,
  $\mathcal{S}_{\mathcal{L}_{T_{v^*}}}$: set of $S_{L_{(v^*, v_j)}}$ for all leaf nodes,
  $v^*$: center node,
   $B$: budget, $flag$: indicate which
greedy strategy is used \;}
   \KwOut{$\mathcal{H}$: the result subset\;}
$\mathcal{H} \leftarrow \emptyset$; $\mathcal{R} \leftarrow \{v^*\}$\;\label{line:functionBegin}
$S_{D^*} \leftarrow $ cell-based dataset of $v^*$\; 
$\mathcal{U} \leftarrow (\cup_{S \in \market'} S) \backslash S_{D^*}$\;

\While{$\mathcal{L}_{T_{v^*}} \neq \emptyset$ and $\sum\limits_{v_j \in \mathcal{R}} p(S_{D_j}) \leq B$}{\label{line:begingreedyH1}
$\Delta_p = p(S_{L_{(v^*,v_j)}})-\sum_{v_i \in (L_{(v^*,v_j)} \cap \mathcal{R})} p(S_{D_i}) $\;
\eIf{flag is 0}{
$v_j = \argmax_{v_j \in \mathcal{L}_{T_{v^*}}} |S_{L_{(v^*,v_j)}} \cap \mathcal{U}|/\Delta_p$ \label{line:ratio}
}{
$v_j = \argmax_{v_j \in \mathcal{L}_{T_{v^*}}} |S_{L_{(v^*,v_j)}} \cap \mathcal{U}|$ \label{line:coverage}
}

\label{line:maxRatio}
\If{$\sum\limits_{v_i \in \mathcal{R}} p(S_{D_i}) + \Delta_p \leq B$}{
$\mathcal{R} \leftarrow \mathcal{R} \cup L_{(v^*,v_j)}$\;
$\mathcal{U} \leftarrow \mathcal{U}\backslash S_{L_{(v^*,v_j)}}$\;
}
$\mathcal{L}_{T_{v^*}} \leftarrow \mathcal{L}_{T_{v^*}} \backslash \{v_j\}$\;

}\label{line:endgreedyH1}
\ForEach{$v_i \in \mathcal{R}$}{
$\mathcal{H} \leftarrow \mathcal{H} \cup \{S_{D_i}\}$\;
}
} 

\Return $\mathcal{H}$\; \label{line:functionEnd}
\end{algorithm}

From the above analysis, we can see that \SG exhibits a favorable approximate ratio of $\frac{p_{min}}{B}$ for a small budget $B$. However, the approximation ratio diminishes when confronted with a large budget $B$. To address this, we propose a \emph{Dual-Path Search Algorithm} (called \BGPA) to provide a superior approximation ratio for scenarios involving large $B$, inspired by the algorithm that addresses the connected maximum coverage problem~\cite{Vandin2011,Vandin2012}.

In the next subsection, we will describe the details of \BGPA.
Essentially, we try to find a collection of paths with a common node in the spatial dataset graph (defined in Def.~\ref{def:datasetGraph}), where the path is represented as a sequence of nodes connected by edges. This ensures that the subset found is a connected subgraph.
However, finding a combination of paths with a common node that maximizes the approximation ratio can be time-consuming (see the time complexity analysis in Section~\ref{subsec:BGPA}). To improve the efficiency, we propose two accelerate strategies that produce competitive solutions with \BGPA in practice while achieving significant speedups in Section~\ref{subsec:accelerate}.

\subsubsection{Dual-Path Search Algorithm (\BGPA)}  
\label{subsec:BGPA}
In the subsequent subsections, we will introduce the definitions and notations related to \BGPA, followed by the implementation details of Algorithm~\ref{alg:IBGPA}. Afterwards, we will further prove the approximation ratio of the algorithm and analyze its time complexity.

\begin{definition1}
\label{def:shortestPath}
    (\textbf{Shortest Path (Distance)}) The shortest path $\mathcal{SP}(v_i, v_j)$ between two nodes $v_i \in \mathcal{V}$ and $v_j \in \mathcal{V}$  in a spatial dataset graph $\mathcal{G}_{\market, \delta}= (\mathcal{V}, \mathcal{E})$  is the shortest sequence of non-repeated and ordered nodes $\{v_i, \dots, v_j\}$ connected by edges present in the spatial dataset graph $\mathcal{G}_{\market, \delta}$. The shortest path distance $|\mathcal{SP}(v_i, v_j)|$ is the number of edges in the shortest path connecting them.
\end{definition1}

\begin{definition1}
(\textbf{Eccentricity}) 
The eccentricity $\xi(v_i, \mathcal{G}_{\mathcal{S}_\mathcal{D}, \delta})$ of a given node $v_i \in \mathcal{V}$ in spatial dataset graph $\mathcal{G}_{\mathcal{S}_\mathcal{D}, \delta} = (\mathcal{V}, \mathcal{E})$ is the greatest distance between $v_i$ and any node $v_j \in \mathcal{V}$, i.e., $\xi(v_i, \mathcal{G}_{\mathcal{S}_\mathcal{D}, \delta}) = \max_{v_j \in \mathcal{V}}|\mathcal{SP}(v_i, v_j)|$.  For ease of presentation, we omit the input parameter $\mathcal{G}_{\mathcal{S}_\mathcal{D}, \delta}$ when the context is clear.
\end{definition1}

\begin{definition1}
\label{def:diameter}
    (\textbf{Diameter}) 
    The diameter $\sigma(\mathcal{G}_{\mathcal{S}_\mathcal{D}, \delta})$ of a spatial dataset graph $\mathcal{G}_{\mathcal{S}_\mathcal{D}, \delta} = (\mathcal{V}, \mathcal{E})$ is the maximum eccentricity of any node in the graph, i.e., $\sigma(\mathcal{G}_{\mathcal{S}_\mathcal{D}, \delta}) = \max_{v \in \mathcal{V}} \xi(v)$. For ease of presentation, we omit the input parameter $\mathcal{G}_{\mathcal{S}_\mathcal{D}, \delta}$ when the context is clear.
\end{definition1}


\begin{definition1}
\label{def:centerRadius}
    (\textbf{Raidus and Center}) 
    The radius $r$ of a spatial dataset graph $\mathcal{G}_{\mathcal{S}_\mathcal{D}, \delta} = (\mathcal{V}, \mathcal{E})$ is the minimum eccentricity among all nodes in $\mathcal{V}$, i.e., $r = \min_{v \in \mathcal{V}} \xi(v)$. The node $v^*$ with the minimum eccentricity is called the center of the graph, i.e., $\xi(v^*) = r$.
\end{definition1}

\begin{definition1}
\label{def:BFSTree}
    (\textbf{BFS tree}) The BFS tree $T_{v^*}$ rooted at the center $v^*$ of the spatial dataset graph $\mathcal{G}_{\mathcal{S}_\mathcal{D}, \delta}$ is generated by running the breadth-first search (BFS), which starts from the selected node $v^*$ and traverses the graph layerwise, thus forming a tree structure based on the order in which nodes are discovered.
    
\end{definition1}

To facilitate the understanding of \BGPA, we first introduce some notions related to the BFS tree.
The notation $\mathcal{L}_{T_{v^*}}$ represents the collection of all leaf nodes of the BFS tree $T_{v^*}$. $L_{(v^*,v_j)}$ represents the collection of nodes on the shortest path from $v^*$ to any node $v_j \in \mathcal{L}_{T_{v^*}}$ (including $v_j$ but not $v^*$), i.e., $L_{(v^*,v_j)} = \mathcal{SP}(v^*, v_j)\backslash v^*$. Furthermore, $S_{L_{(v^*,v_j)}}$ represents the union of cell-based datasets associated with nodes in the path $L_{(v^*,v_j)}$, i.e., $S_{L_{(v^*,v_j)}} = \cup_{v_j \in L_{(v^*,v_j)}} S_{D_j}$. $p(S_{L_{(v^*,v_j)}})$ represents the price of the path $L_{(v^*,v_j)}$, which is the sum of the prices of cell-based datasets corresponding to all nodes on the path $L_{(v^*,v_j)}$, i.e., $p(S_{L_{(v^*,v_j)}}) = \sum_{v_j \in L_{(v^*,v_j)}} p(S_{D_j})$. $\mathcal{S}_{\mathcal{L}_{T_{v^*}}}$ represents the collection of $S_{L_{(v^*, v_j)}}$ for each $v_j \in \mathcal{L}_{T_{v^*}}$, i.e., $\mathcal{S}_{\mathcal{L}_{T_{v^*}}} = \{S_{L_{(v^*,v_j)}}, v_j \in \mathcal{L}_{T_{v^*}}\}$. $\mathcal{H}$ is used to record the union of nodes on the selected path. In the following, we describe the details of \BGPA according to the pseudocode of Algorithm~\ref{alg:IBGPA}.

\noindent \textbf{Step 1: Construction of Spatial Dataset Graph.}
We begin by filtering datasets in the data marketplace $\market$ that exceed the allocated budget, resulting in a subset of candidate datasets $\market'$ (see Lines~\ref{line:DPSAFliter1} to \ref{line:DPSAFliter2}).
We then construct the spatial dataset graph $\mathcal{G}_{\market', \delta} = (\mathcal{V},\mathcal{E}))$ to support the execution of subsequent steps (see Line~\ref{line:constructGraph}). Specifically, for each node $v_i \in \mathcal{V}$, we first compute the cell-based dataset distance (defined in Def.~\ref{defi:dist}) between $v_i$ and each node in $\mathcal{V}$. If the distance does not exceed the distance threshold $\delta$, we add an edge to the spatial dataset graph $\mathcal{G}_{\market', \delta}$. The dataset graph construction process continues until all connected dataset pairs are identified.
 
\noindent \textbf{Step 2: Finding Connected Subgraphs.}
Given the distribution of datasets across diverse locations on Earth, the constructed spatial data graph is likely not a connected graph. To mitigate the risk of the algorithm falling into a local optimal solution, we need to execute the algorithm for each connected subgraph $\mathcal{G}_{\market', \delta}^i \in \mathcal{G}_{\market', \delta}$. Therefore, it is necessary to identify the set $\mathcal{G}_{\market', \delta}$ comprising all connected subgraphs subsequent to the dataset graph construction (see Line~\ref{line:eachSubgraph}). The process of finding a subgraph involves randomly selecting a node and employing BFS within the dataset graph to find all nodes within a connected subgraph. This process repeats for the remaining nodes until all nodes have been traversed.

\noindent \textbf{Step 3: Finding Center of Subgraph.}
Next, our objective is to find the center and the radius for each connected subgraph $\mathcal{G}_{\market', \delta}^i$. To achieve this, we employ BFS rooted at each node $v \in \mathcal{V}$ and obtain a BFS tree $T_{v}$, where the depth of $T_{v}$ is the eccentricity of the node $v$. Following this procedure for all nodes, we can identify the center $v^*$ with the minimum eccentricity. In the following, we implement a two-round search to find a better feasible solution.

\myparagraph{Step 4. Generating Feasible Solutions}  
In the first-round search, the \texttt{BudgetedGreedy} function (see Lines~\ref{line:functionBegin} to \ref{line:functionEnd}) selects the leaf node $v_j$ from $\mathcal{L}_{T_{v^*}}$ with the maximum ratio $|S_{L_{(v^*,v_j)}} \cap \mathcal{U}|/\Delta_p$ in each iteration (see Line~\ref{line:ratio}). Then, the function returns the result subset, which is a set of datasets corresponding to all nodes along the selected paths, until the candidate set traversal ends or the remaining budget is exhausted. Similarly, to further prevent the algorithm from falling into a local optimal solution, the function is called again to obtain an alternative feasible solution.
This is accomplished by selecting the leaf node $v_j$ with the maximum spatial coverage increment $|S_{L_{(v^*,v_j)}} \cap \mathcal{U}|$ in each iteration (see Line~\ref{line:coverage}). Finally, the algorithm returns the better result subset $\mathcal{H}$ from two feasible solutions.

\begin{figure*}[t]
\setlength{\belowcaptionskip}{0 cm}	
\centering
\includegraphics[width=0.98\textwidth]{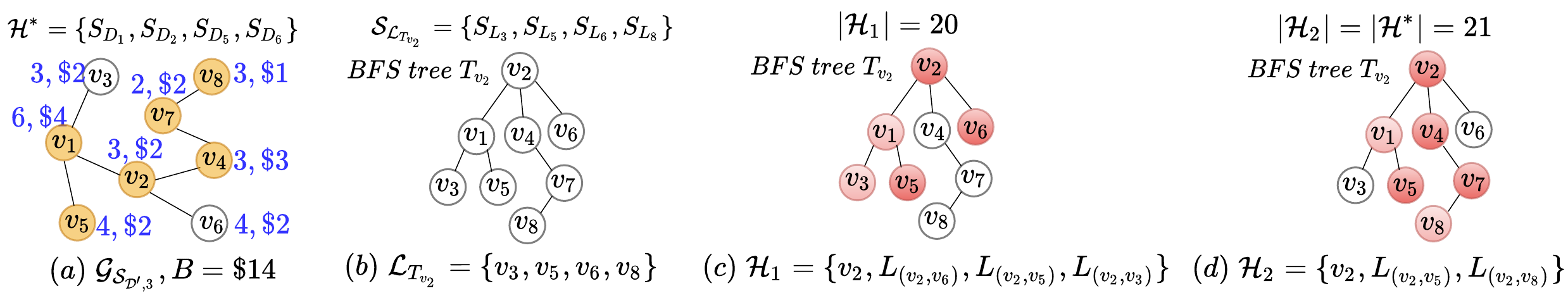}
\caption{Illustration of \BGPA, where (a) shows the optimal solution $\mathcal{H}^*$ under $B = \$14$; (b) shows the BFS tree $T_{v_2}$ rooted at the center $v_2$, (c) shows the solution $\mathcal{H}_1$ found by the first-round search, and (d) shows the solution $\mathcal{H}_2$ found by the second-round search.
}
\label{fig:greedyPath}
\end{figure*}
\begin{example1}
\label{exa:BGPA}
To facilitate the understanding of \BGPA, we use an example similar to Fig.~\ref{fig:simpleGreedy} for illustration.
Fig.~\ref{fig:greedyPath}(a) shows the optimal solution of \MCPBC in $\market'$ when $B$ is increased to $\$14$. Algorithm~\ref{alg:IBGPA} first performs the pre-processing step, which runs BFS on each node $v$ and calculates its eccentricity $\xi(v)$.  Notably, node $v_2$ exhibits a minimum eccentricity of 3. Thus, $v_2$ is selected as the center and $3$ as the radius $r$. Fig.~\ref{fig:greedyPath}(b) shows the BFS tree $T_{v_2}$ rooted at $v_2$.
All leaf nodes form a subset $\mathcal{L}_{T_{v_2}}$=$\{v_3, v_5, v_6, v_8\}$. For each node $v_j$ $\in$ $\mathcal{L}_{T_{v_2}}$, the union of the cell-based datasets of all nodes on the shortest path $L_{(v_2, v_j)}$ forms a new set $\mathcal{S}_{\mathcal{L}_{T_{v_2}}}$=$\{ S_{L_{(v_2, v_3)}}, S_{L_{(v_2, v_5)}}, S_{L_{(v_2, v_6)}}, S_{L_{(v_2, v_8)}}\}$.

During the first function call, \textsf{DPSA} first adds the node $v_2$ to the result subset $\mathcal{H}_1$ and continues its execution because $B- p(S_{D_2})$=\$14-\$2=\$12$\geq r\times p_{max}$. Next, \textsf{DPSA} iteratively selects the path with the maximum ratio {$|S_{L_{(v^*,v_j)}} \cap \mathcal{U}|/\Delta_p$ ($v_j \in \mathcal{L}_{T_{v^*}}$)} among all available paths. As shown in Fig.~\ref{fig:greedyPath}(c), \textsf{DPSA} adds the path $L_{v_2, v_6}$ to $\mathcal{H}_1$, as it has the maximum ratio of $4/2=2$. The remaining budget is updated to $\$12-\$2 = \$10$. Next, the ratios for the remaining three paths are calculated as 9/6, 10/6, and 8/6. Thus, \textsf{DPSA} adds the path $L_{v_2, v_5}$ to $\mathcal{H}_1$. The budget is then reduced to $\$10-\$6 = \$4$. Since the price of the path $L_{(v_2, v_8)}$ is $\$6$, which exceeds the remaining budget, the function instead selects $L_{(v_2, v_3)}$ and adds it to $\mathcal{H}_1$, after which the function terminates.

During the second function call, \textsf{DPSA} iteratively selects the path with the maximum spatial coverage increment {$|S_{L_{(v^*,v_j)}} \cap \mathcal{U}|$ ($v_j \in \mathcal{L}_{T_{v^*}}$)} among all available paths. \textsf{DPSA} first adds the node $v_2$ to the result subset $\mathcal{H}_2$. Next, the spatial coverage increment $|S_{L_{(v_2,v_j)}} \cap \mathcal{U}|$ for the four paths from the root node $v_2$ to leaf nodes $v_3, v_5, v_8$ and $v_6$ are 9, 10, 8, and 4, respectively.
Therefore, \textsf{DPSA} selects $L_{(v_2, v_5)}$, which offers the maximum coverage increment $|S_{L_{(v_2, v_5)}} \cap \mathcal{U}|$ = 10. 
After selecting the path $L(v2, v5)$, the remaining budget is calculated as $12-p(S_{D_1})-p(S_{D_5}) = \$12-\$4-\$2 = \$6$.

The function then proceeds with the next iteration. The coverage increments of the remaining three paths from $v_2$ to $v_3, v_8$ and $v_6$ are 3, 8, and 4, respectively, with corresponding costs of \$2, \$6, and \$2. Thus, \textsf{DPSA} selects the path $L(v_2, v_8)$, which offers the maximum coverage increment 8. At this point, the budget is completely exhausted, and the function terminates, yielding the result set $\mathcal{H}_2 = \{v_2, L_{(v_2, v_5)}, L_{(v_2, v_8)}\}$, as shown in Fig.~\ref{fig:greedyPath}(d). Finally, since the spatial coverage $|\mathcal{H}_1| = 20 < |\mathcal{H}_2|=21$, the algorithm returns $\mathcal{H}_2$, which is identical to the optimal solution $\mathcal{H}^*$.

\end{example1}

\noindent \textbf{Theoretical Analysis}. 
During the execution of \BGPA, we need to select the path $S_{L_{(v^*, v_j)}}$ to the result subset $\mathcal{H}$ at each step. Consequently, it is crucial for the provided budget to be substantial enough to carry out at least one iteration; otherwise, the obtained result subset $\mathcal{H}$ may be empty, lacking any approximation guarantees. To ensure that the algorithm is theoretically guaranteed on a budget $B$, we make the reasonable assumption that the budget $B \geq r \times p_{max}+ p(S_{D^*})$.
In the ensuing analysis, we will first prove Lemma~\ref{lemma:iterationNumber}, which is then used to prove the approximation ratio of DPSA in Theorem~\ref{lemma:BGPA}.
Let $k'$ denote the minimum number of iterations in the algorithm within the budget $B$, and $k$ denote the size of the optimal result subset in \MCPBC. Then we have the following lemma.

\begin{lemma}
\label{lemma:iterationNumber}
The minimum number of iterations of the greedy search $k'$ is at least $\frac{p_{min}}{r \times p_{max}}(1-\frac{p_{max}}{B})k$.
\end{lemma}

\begin{proof}

The price of the dataset corresponding to center node $v^*$ is denoted as $p(S_{D^*})$, which is not greater than $p_{max}$. Considering that the eccentricity of the center node $v^*$ is the radius $r$, the number of nodes in $L_{(v^*, v_j)}$ ($v_j$ $\in$ $\mathcal{L}_{T_{v^*}})$ is less than $r$. Then the \texttt{BudgetedGreedy} function costs at most $r \times p_{max}$ budget at each iteration. Thus, we can establish the inequality $ \frac{B-p(S_{D^*})}{r \times p_{max}} \leq k'\leq \frac{B-p(S_{D^*})}{r \times p_{min}}$. Since $k$ is number of nodes in the optimal solution of \MCPBC, $\frac{B}{p_{max}} \leq k \leq \frac{B}{p_{min}}$. 
We assume that $k' \geq x \times k$, where $x$ is the coefficient to be found. We have
\vspace{-0.1cm}
\begin{equation}
\small
\begin{aligned}
 x \times \frac{B}{p_{max}}   \leq x \times k \leq x \times \frac{B}{p_{min}}. \\
\end{aligned}
\end{equation}

Since we know that $k'\geq x \times k$, the lower bound of $k'$ must be no less than the upper bound of $x \times k'$, i.e., $\frac{B-p(S_{D^*})}{r \times p_{max}} \geq x \times \frac{B}{p_{min}}$. Then,\\
\vspace{-0.1cm}
\begin{small}
\begin{equation}
\begin{aligned}
    x \leq \frac{p_{min}}{r \times p_{max}}(1- \frac{p(S_{D^*})}{B}).
\end{aligned}
\end{equation}
\end{small}
Let $x$ be $\frac{p_{min}}{r \times p_{max}}(1- \frac{p(S_{D^*})}{B})$, we have
\begin{small}
\begin{equation}
\hspace*{-0.3cm} 
\begin{aligned}
k'&\!\geq \frac{p_{min}}{r \times p_{max}}(1- \frac{p(S_{D^*})}{B}) k \geq \frac{p_{min}}{r \times p_{max}}(1- \frac{p_{max}}{B}) k. \\
\end{aligned}
\end{equation}
\end{small}

\end{proof}

\noindent Next, we analyze the approximation ratio of \BGPA.
\begin{theorem}
\label{lemma:BGPA}
    \BGPA delivers a $\frac{p_{min}}{r \times p_{max}}(1-\frac{p_{max}}{B}) (1-\frac{1}{e})$ approximate solution for the \MCPBC.
\end{theorem}

\begin{proof}
(1) We first analyze the approximation ratio of the second function call of Algorithm~\ref{alg:IBGPA}.
We define $\mathcal{H}_{k'}$ as the output after $k'$ iterations during the first function call.
Additionally, let $\mathcal{H}_{path}'$ be the optimal solution of the maximum $k'$-coverage problem \cite{Hochbaum1998} from $\mathcal{S}_{\mathcal{L}_{T_{v^*}}}$.
We observe that $\mathcal{H}_{k'}$ can be regarded as the output of the greedy algorithm~\cite{Hochbaum1998} for the maximum $k'-$coverage problem. This is because all paths $L_{(v^*, v_j)} (v_j\in \mathcal{L}_{T_{v^*}})$ have the common root node $v^*$, which forms a connected subgraph, and the total price is less than the budget $B$. 
For ease of presentation, we use $|\mathcal{H}_{k'}| = |\bigcup_{S \in \mathcal{H}} S|$ to represent the spatial coverage of $\mathcal{H}_{k'}$.
By Theorem 1 in the maximum coverage problem \cite{Hochbaum1998}, we have
\vspace{-0.1cm}
\begin{equation}
\small
    |\mathcal{H}_{k'}| > (1-\frac{1}{e})  |\mathcal{H}_{path}'|.
\end{equation}

Let $\mathcal{H}_{node}'$ be the optimal solution of maximum $k'$-coverage problem from $\market'$. 
Since for each dataset $S_{D_i}\in \market'$, there is some subsets $S_{L_{(v^*, v_j)}}\in \mathcal{S}_{\mathcal{L}_{T_{v^*}}}$ such that $S_{D_i} \subseteq S_{L_{(v^*, v_j)}}$. Thus the optimal solution of $k'$ subsets from $\mathcal{H}_{path}'$ is better than that from $\mathcal{H}_{node}'$, that is

\begin{equation}
\small
    |\mathcal{H}_{path}'| \geq |\mathcal{H}_{node}'|.
\end{equation}

Let $\mathcal{H}_{node}$ be the optimal solution of the maximum $k-$coverage problem from $\market'$. Then $\mathcal{H}_{node}' \geq \frac{k'}{k} \mathcal{H}_{node}$. In addition, we have known that $k' \geq \frac{p_{min}}{r \times p_{max}}(1- \frac{p_{max})}{B}) k$. Thus, we have
\begin{equation}
\small
    \mathcal{H}_{node}' \geq \frac{k'}{k} \mathcal{H}_{node} \geq \frac{p_{min}}{r \times p_{max}}(1-\frac{p_{max}}{B}) \mathcal{H}_{node}.
\end{equation}

Let $\mathcal{H}^*$ be the optimal solution of \MCPBC from $\market'$. We have

\begin{equation}
\small
   |\mathcal{H}_{node}|  \geq |\mathcal{H}^*|.
\end{equation}

Combining the above Inequalities, we have
\begin{equation}
\small
\begin{aligned}
       |\mathcal{H}_{k'}| &\; \geq (1-\frac{1}{e})  |\mathcal{H}_{path}'|\\
       &\; \geq (1-\frac{1}{e})|\mathcal{H}_{node}'| \\
       &\; \geq \frac{p_{min}}{r \times p_{max}}(1-\frac{p_{max}}{B}) (1-\frac{1}{e})|\mathcal{H}_{node}|\\
        &\; \geq \frac{p_{min}}{r \times p_{max}}(1-\frac{p_{max}}{B}) (1-\frac{1}{e})|\mathcal{H}^*|.
\end{aligned}
\end{equation}

(2) The first function call is to iteratively select the node with the maximum ratio $|S_{L_{(v^*,v_j)}} \cap \mathcal{U}|/\Delta_p$. Although it cannot provide any approximation guarantees in theory, but are highly effective in practice, as confirmed by our experiments. In summary, \BGPA provides a $\frac{p_{min}}{r \times p_{max}}(1-\frac{p_{max}}{B}) (1-\frac{1}{e})$-approximation for the \MCPBC.

\end{proof}

Theorem~\ref{lemma:BGPA} demonstrates the approximation achieved by the algorithm when executed at the center node $v^*$ with a radius $r = \xi(v^*)$ in the graph (see Line~\ref{line:findCenter}).  
If we random choose a node $v$ as the center to perform Algorithm~\ref{alg:IBGPA}, the approximation ratio is $\frac{p_{min}}{\xi(v) \times p_{max}}(1-\frac{p_{max}}{B})k$. By observing the approximation ratio, we can conclude that a smaller value of $\xi(v)$ leads to a higher ratio for \BGPA. 
This is precisely why we choose the center $v^*$ with the minimum eccentricity $\xi(v^*)$, as it ensures the maximum guaranteed approximation ratio. In the upcoming discussion, we show that when $B$ exceeds a certain value, \BGPA provides a higher approximation ratio than \SG.

\begin{lemma}
    For $B \geq \frac{r \times e \times p_{max}}{e-1}+p_{max}$, \BGPA delivers a better approximation ratio than \SG, i.e., $\frac{p_{min}}{r \times p_{max}}(1-\frac{p_{max}}{B}) (1-\frac{1}{e}) \geq \frac{p_{min}}{B}$.
\end{lemma}

\begin{proof}
\begin{small}
\begin{equation}
\begin{aligned}
 &\; \frac{p_{min}}{r \times p_{max}}(1-\frac{p_{max}}{B}) (1-\frac{1}{e})   \geq \frac{p_{min}}{B} \\
 &\; \Rightarrow \frac{p_{min}}{B} \times \frac{B-p_{max}}{r\times p_{max}} \times \frac{e-1}{e} \geq \frac{p_{min}}{B}\\
&\; \Rightarrow \frac{B-p_{max}}{r\times p_{max}} \times \frac{e-1}{e}  \geq 1\\
&\; \Rightarrow  B \geq \frac{r \times e \times p_{max}}{e-1}+p_{max}
\end{aligned}
\end{equation}
\end{small}
    Thus, \BGPA delivers a better approximation ratio when $B \geq \frac{r \times e \times p_{max}}{e-1}+p_{max}$.
\end{proof}

\myparagraph{Time Complexity} As discussed in the previous subsection, the time complexity of finding the candidate datasets $\market'$ takes $O(n)$, where $n = |\market|$. Then, to establish the graph $\mathcal{G}_{\market', \delta}$, we compute the cell-based dataset distance between each pair of dataset nodes, which requires a total of $n+(n-1)+ \dots+1 = \frac{n(n+1)}{2}$ operations. Subsequently, to find the center of $\mathcal{G}_{\market', \delta}$, we run the BFS for all nodes in the graph. The complexity of BFS is known to be $O(m)$ \cite{Hochbaum2020}, where $m$ is the number of edges in $\mathcal{G}_{\market', \delta}$. Therefore, finding the center of a graph takes at most $O(nm)$. Next, we compute $S_{L_{(v^*, v_j)}}$ for each leaf node $v_j$, which involves performing $O(n)$ set union operations. Each set union operation takes $O(|\mathcal{U}|)$ ($|\mathcal{U}| = |\cup_{S \in \market'} S|$), as we described in the time complexity analysis in Section~\ref{subsec:SG}. Consequently, the time complexity for this step is $O(n|\mathcal{U}|)$, where $|\mathcal{U}| = |\cup_{S \in \market'} S|$. Finally, during two function call of \BGPA, there are at most $2 \times k'$ iterations ($\frac{B-p(S_{D^*})}{r \times p_{max}} \leq k' \leq \frac{B-p(S_{D^*})}{r \times p_{min}}$). Each iteration requires at most $O(n)$ set union operations to pick the path with the maximum marginal gain. In summary, the total time complexity of \BGPA is $O(n+\frac{n(n+1)}{2}+nm+|\mathcal{U}|n+2k'|\mathcal{U}|n) = O((m+k'|\mathcal{U}|)n + n^2)$.

\subsubsection{Acceleration Strategies}
\label{subsec:accelerate}
Based on the time complexity analysis of \BGPA,  it is evident that the construction of the spatial dataset graph and finding the center node are particularly time-consuming processes. Consequently, we design two acceleration strategies to enhance the efficiency of \BGPA. The index-based acceleration strategy mainly focuses on rapidly constructing a dataset graph. When finding all connected dataset pairs for each node and adding edges to the graph, it can employ pruning strategies to batch filter out those datasets that do not satisfy the spatial connectivity. Additionally, the BFS-based acceleration strategy can quickly find the center with twice BFS, without having to perform the BFS for all nodes, significantly reducing the time to find the center. In the following, we describe the details of two acceleration strategies.

\noindent\textbf{Index-based Acceleration.} 
\label{sec:accelerate1}
We first extend the ball tree~\cite{Omohundro1989} to speed up the construction of the spatial dataset graph. Given a set of cell-based datasets $\mathcal{S}_{\mathcal{D}} = \{S_{D_1}, \dots, S_{D_{|\mathcal{D}|}}\}$, each dataset can be treated as a ball node. 
The centroid of the dataset (i.e., the mean of all cell coordinates within the dataset) serves as the center of the ball node, and the radius is determined by the maximum distance between the centroid and any cell coordinates within the dataset. This differs from the concept of the center and radius of the dataset graph in Definition~\ref{def:centerRadius}. We then construct a tree index in a top-down way based on all ball nodes~\cite{Yang2022}.

During the construction of the dataset graph using the tree index, we initiate a graph $\mathcal{G}_{\mathcal{S}_\mathcal{D}, \delta}$. For each node $v_i \in \mathcal{V}$, we perform a top-down recursive search from the root ball to identify all nodes connected with $v_i$. For a ball node (denoted as $v_b$) in the tree visited during a recursive search, if the distance between the centers of $v_i$ and $v_b$ minus the radii of both nodes exceeds the threshold $\delta$, it indicates any child of $v_b$ is also farther than $\delta$ from $v_i$. Thus, $v_b$ and its children can be safely pruned. Conversely, if the sum of the distances between the centers of $v_i$ and $v_b$ plus the radii of both nodes is less than $\delta$, it implies that any child of $v_b$ is closer than $\delta$ to $v_i$, indicating that $v_i$ is directly connected to any child of $v_b$.

Otherwise, we proceed with the search for the left and right child nodes of $v_b$ until a leaf node (denoted as $v_j$) is reached. We then compute the dataset distance between $v_i$ and $v_j$. If $dist(v_i, v_j) \leq \delta$, we add an undirected edge from $v_i$ to $v_j$ in $\mathcal{G}_{\mathcal{S}_\mathcal{D}, \delta}$. This process terminates until all nodes in the graph have been visited.

\noindent\textbf{BFS-based Acceleration.} 
\label{sec:accelerate2}
To speed up finding the center, we adopt a two-round BFS strategy. Specifically, we start a round of BFS by randomly choosing a node $v_i$, which enables us to find the node $v_j$ that is furthest away from $v_i$ (if there are multiple farthest nodes, we randomly choose one). We then take the node $v_j$ as the starting node, redo the BFS, and find the node $v_k$ farthest from $v_j$, where the shortest path distance (defined in Def.~\ref{def:shortestPath}) from $v_j$ to $v_k$ is the diameter of the graph, and the center of the path is one of the centers of the graph. 
Next, we prove the correctness of the above search process in an acyclic graph $\mathcal{G}_{\mathcal{S}_\mathcal{D}, \delta}=(\mathcal{V}, \mathcal{E})$. 

We use $v_s$ and $v_e$ to represent two nodes, where the shortest path distance $|\mathcal{SP}(v_s, v_e)|$ from $v_s$ to $v_e$ is equal to the diameter of $\mathcal{G}_{\mathcal{S}_\mathcal{D}, \delta}$, i.e., $|\mathcal{SP}(v_s, v_e)| = \sigma(\mathcal{G}_{\mathcal{S}_\mathcal{D}, \delta})$. Then, we have the following theorem.  
\begin{theorem}
 $\forall v_i \in \mathcal{V}$, the node $v_j$ that is farthest from $v_i$ must be $v_s$ or $v_e$, i.e., $v_j \in \{v_s, v_e\}$.
\end{theorem}

\begin{proof}
   For the node $v_j$ farthest from the given node $v_i$, there are three possible cases in total. In the following, we will prove that the first two cases are not true, and only the third case holds true. 

 \noindent \textbf{Case 1:} $v_j \notin \{v_s, v_e\}$ and the $\mathcal{SP}(v_i, v_j)$ at least has a common node $v_c$ with $\mathcal{SP}(v_s, v_e)$. Due to $v_c \in \mathcal{SP}(v_s, v_e)$, the node farthest from $v_c$ must be $v_s$ or $v_e$. Then, the farthest node found after passing $v_c$ from $v_i$ must be $v_s$ or $v_e$. Thus, this case is not true.

\noindent \textbf{Case 2:} $v_j \notin \{v_s, v_e\}$ and the $\mathcal{SP}(v_i, v_j)$ has no common node with $\mathcal{SP}(v_s, v_e)$. We use $v_t$ to represent one of the nodes in $\mathcal{V}$.
Since the node $v_j$ is the furthest node from $v_i$ and $v_j \notin \{v_s, v_e\}$, we have 
\begin{equation}
\nonumber
\hspace{-0.1cm}
\fontsize{9.5pt}{10pt}\selectfont
\begin{aligned}
       &\; \quad \:\:\, |\mathcal{SP}(v_i, v_j)| > |\mathcal{SP}(v_i, v_e)| \\
       &\; \Rightarrow  |\mathcal{SP}(v_i, v_j)| > |\mathcal{SP}(v_i, v_t)|+|\mathcal{SP}(v_t, v_e)| \\
       &\; \Rightarrow |\mathcal{SP}(v_i, v_j)| +|\mathcal{SP}(v_i, v_t)|  > |\mathcal{SP}(v_t, v_e)| \\
       &\; \Rightarrow |\mathcal{SP}(v_i, v_j)| +|\mathcal{SP}(v_i, v_t)| + |\mathcal{SP}(v_s, v_t)|> \\ &\;\quad \quad |\mathcal{SP}(v_t, v_e)| + |\mathcal{SP}(v_s, v_t)| \\
       &\; \Rightarrow |\mathcal{SP}(v_s, v_j)|  > |\mathcal{SP}(v_s, v_e)| 
\end{aligned}
\end{equation}

However, it is impossible to have a shortest path distance that is greater than the diameter. Thus, this case is also not true. 

\noindent \textbf{Case 3:} $v_j \in \{v_s, v_e\}$. Since the first two cases are false, the \textbf{Case 3} is true.
In summary, the theorem holds.

\end{proof}

\begin{theorem}
Given a node $v_j \in \{v_s,v_e\}$, the shortest path distance $\mathcal{SP}(v_j, v_k)$ from $v_j$ to the farthest node $v_k \in \mathcal{V}$ is equal to the diameter $\sigma$ of the acyclic graph $\mathcal{G}_{\mathcal{S}_\mathcal{D}, \delta}$, i.e., $|\mathcal{SP}(v_j, v_k)| = \sigma$. 
\end{theorem}
\begin{proof}
For the shortest path distance $|\mathcal{SP}(v_j, v_k)|$, there are three possible cases in relation to the diameter.  In the following, we will prove the first two cases are not true, and only the third case holds true. 

\noindent \textbf{Case 1:} $|\mathcal{SP}(v_j, v_k)| > \sigma$. Due to the diameter being the maximum eccentricity of any node in the graph (defined in Def.~\ref{def:diameter}), there cannot be a shortest path such that the distance $|\mathcal{SP}(v_s, v_e)|$ is greater than the diameter. Thus, this case is not true.

\noindent \textbf{Case 2:} $|\mathcal{SP}(v_j, v_k)| < \sigma$. Since the node $v_j \in \{v_s, v_e\}$ and the shortest path distance $|\mathcal{SP}(v_s, v_e)|$ is the greatest distance in $\mathcal{G}_{\mathcal{S}_\mathcal{D}, \delta}$, which equals the diameter of $\mathcal{G}_{\mathcal{S}_\mathcal{D}, \delta}$, the node $v_k$ farthest from $v_j$  has to be one of the nodes in $\{v_s, v_e\}$. Then, the distance $|\mathcal{SP}(v_j, v_k)|$ cannot be less than $\sigma$. Thus, this case is not true. 

\noindent \textbf{Case 3:} $|\mathcal{SP}(v_j, v_k)| = \sigma$.  Since the first two of the three
cases are false, the \textbf{Case 3} is true. In summary, the theorem holds.
\end{proof}
In light of the preceding proof, we know that the node $v_j \in \{v_s, v_e\}$. Then, the shortest path distance $|\mathcal{SP}(v_j, v_k)|$ from $v_j$ to the farthest node $v_k$ must be the diameter of the $\mathcal{G}$. 
Although the aforementioned strategy for determining the center may not achieve absolute accuracy, it proves to be highly effective in practice, as validated by our experimental results.

\end{sloppypar}

\vspace{-0.3cm}
\section{Experiments}
\label{sec:exp}

\begin{sloppypar}

\begin{table*}
\centering
\caption{An overview of five spatial data collections.}
  \label{tab:dataset}
  \scalebox{1.0}{
\renewcommand\arraystretch{1.2}\addtolength{\tabcolsep}{0pt}
		\begin{tabular}{ccccc}
			\hline
			\textbf{Data Collection}     &\textbf{Storage (GB)} & \textbf{Number of Datasets} & \textbf{Number of Points} & \textbf{Coordinates Range}   \\
			\hline
   \texttt{BTAA}      &3.38     & 3,204 & 96,788,280 & [$(-179^{\circ}.46', -87^{\circ}.70'), (179^{\circ}.98', 71^{\circ}.40')$]  \\
			 \texttt{Trackable}     &0.68       &10,000 & 40,442,724 & [$(-53^{\circ}.16', -89^{\circ}.99'), (69^{\circ}.70', 88^{\circ}.21')$]  \\ 
			\texttt{Identifiable}   &0.85     &10,000 & 67,688,576 & [$(-89^{\circ}.68', -89^{\circ}.99'), (88^{\circ}.88', 89^{\circ}.99')$]  \\
			\texttt{Public}      &0.54      & 10,000 & 32,407,907 &[$(-54^{\circ}.85', -89^{\circ}.99'), (79^{\circ}.44', 76^{\circ}.38')$]  \\
   \texttt{OSM}      &4.07      & 60,000 & 207,879,164 &[$(-68^{\circ}.25', -89^{\circ}.99'), (70^{\circ}.06', 89^{\circ}.99')$]  \\
			\hline
		\end{tabular}
  }
  \vspace{-0.3cm}
\end{table*}

In this section, we conduct experiments on real-world dataset collections to demonstrate the effectiveness and efficiency of our proposed algorithms. In Section~\ref{sec:setup}, we introduce the experimental setups. Subsequently, we present the experimental results that vary with several key parameters to demonstrate our algorithms' effectiveness (Section~\ref{sec:effect}) and efficiency (Section~\ref{sec:efficient}), followed by our visualization study in Section~\ref{sec:case}.

\subsection{Experimental Setups}
\label{sec:setup}
\noindent\textbf{Data Collections.} 
We conduct experiments in five real-world spatial data collections. Table~\ref{tab:dataset} shows the statistics of five data collections, which covers the storage, the number of datasets, the number of points, and the range of coordinates. 
The \textsf{BTAA} is sourced from the Big Ten Academic Alliance Geoportal, which offers spatial datasets for the midwestern US, such as bus routes, vehicle radar, and census data.
The other four collections are public datasets we collected from OpenStreetMap, which contain real spatial datasets from the world uploaded by users and marked by a tag (e.g., trackable, identifiable, and public). The \textsf{Trackbale} contains the movement tracks of anonymous entities labelled trackable that can be tracked by satellites or sensors, the \textsf{Identifiable} contains the movement tracks of identifiable individuals labelled identifiable, such as vehicles and bikes, and the \textsf{Public} contains the publicly available location datasets labelled public, such as stations and buildings. \textsf{OSM} contains large-scale datasets of three tags simultaneously to verify the scalability of our algorithms.
Fig.~\ref{fig:distribution} shows the distribution of dataset scale in five dataset collections. We can observe that the dataset distribution is relatively uniform across the different scales, ensuring a comprehensive evaluation across various scenarios.

\smallskip
\noindent\textbf{Parameter Settings.}
The settings of parameters are summarized in Table~\ref{tab:parameter} where the default settings are underlined. Below is the description of each parameter:

\begin{table}
\centering
\caption{Parameter settings.} 

\label{tab:parameter} 
  \scalebox{1.0}{
\renewcommand\arraystretch{1.2}\addtolength{\tabcolsep}{1pt}
		\begin{tabular}{lc}
			\hline
			\textbf{Parameter} & \textbf{Settings}   \\
			\hline
			$B$: budget & \{0.001, 0.005, 0.01, 0.05, \underline{0.1} \}\\ 
            $\delta$: distance threshold & \{0, 5,  \underline{10}, 15, 20 \}\\
			$\theta$: resolution & \{9, 10, \underline{11}, 12,  13\}\\
            $m$: number of datasets & \{0.2, 0.4, 0.6, 0.8, \underline{1}\}$\times |\mathcal{D}|$\\
			\hline
		\end{tabular}
  }
\end{table}

\begin{figure*}[htbp]
\centering
\setlength{\abovecaptionskip}{-0.1 cm} 
\setlength{\belowcaptionskip}{-0.4 cm} 
\subfigure[]{
\begin{minipage}[t]{0.20\linewidth}
\centering
\includegraphics[width=3.8cm,height=2.65cm]{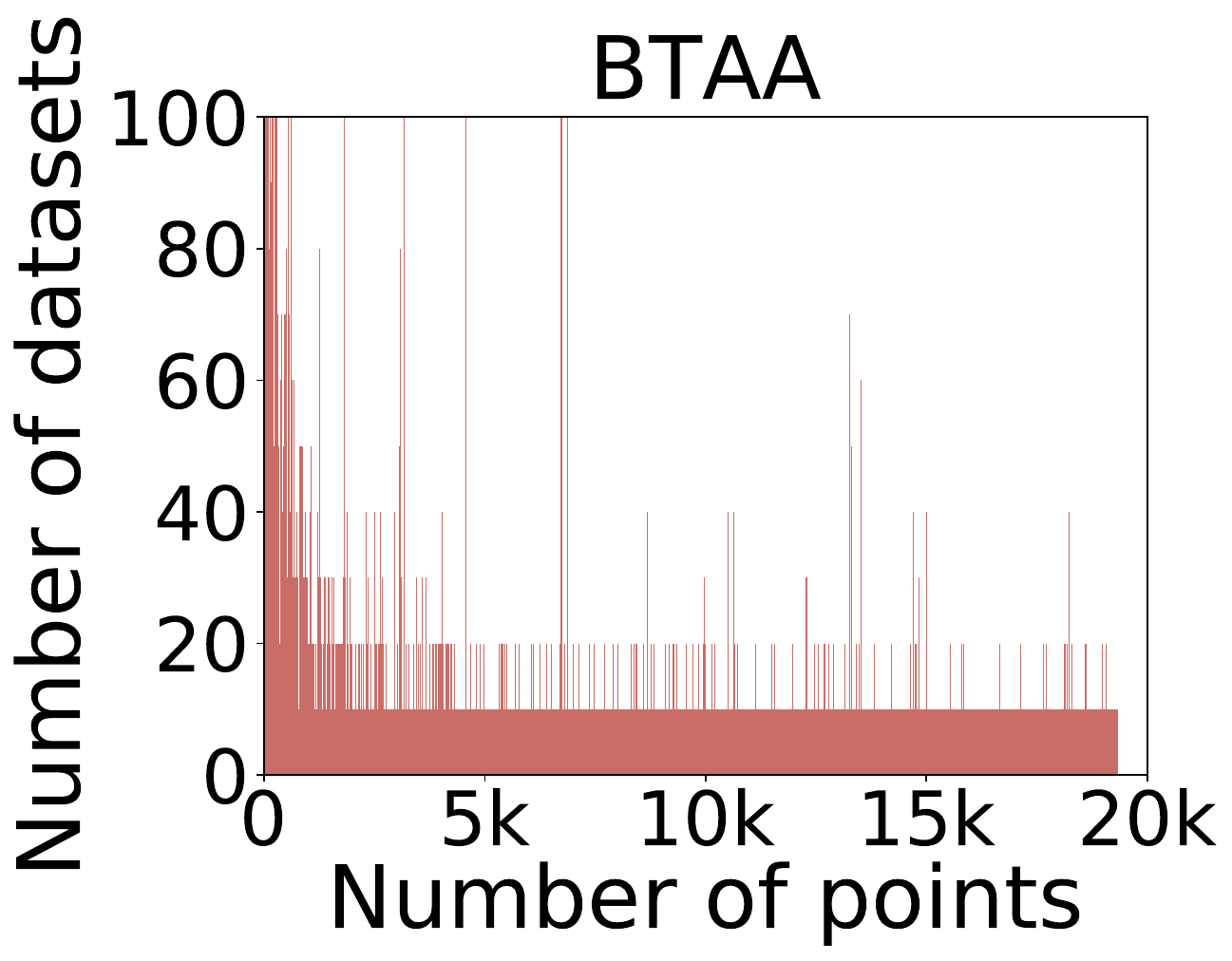}
\end{minipage}%
}%
\subfigure[]{
\begin{minipage}[t]{0.19\linewidth}
\centering
\includegraphics[width=3.2cm,height=2.6cm]{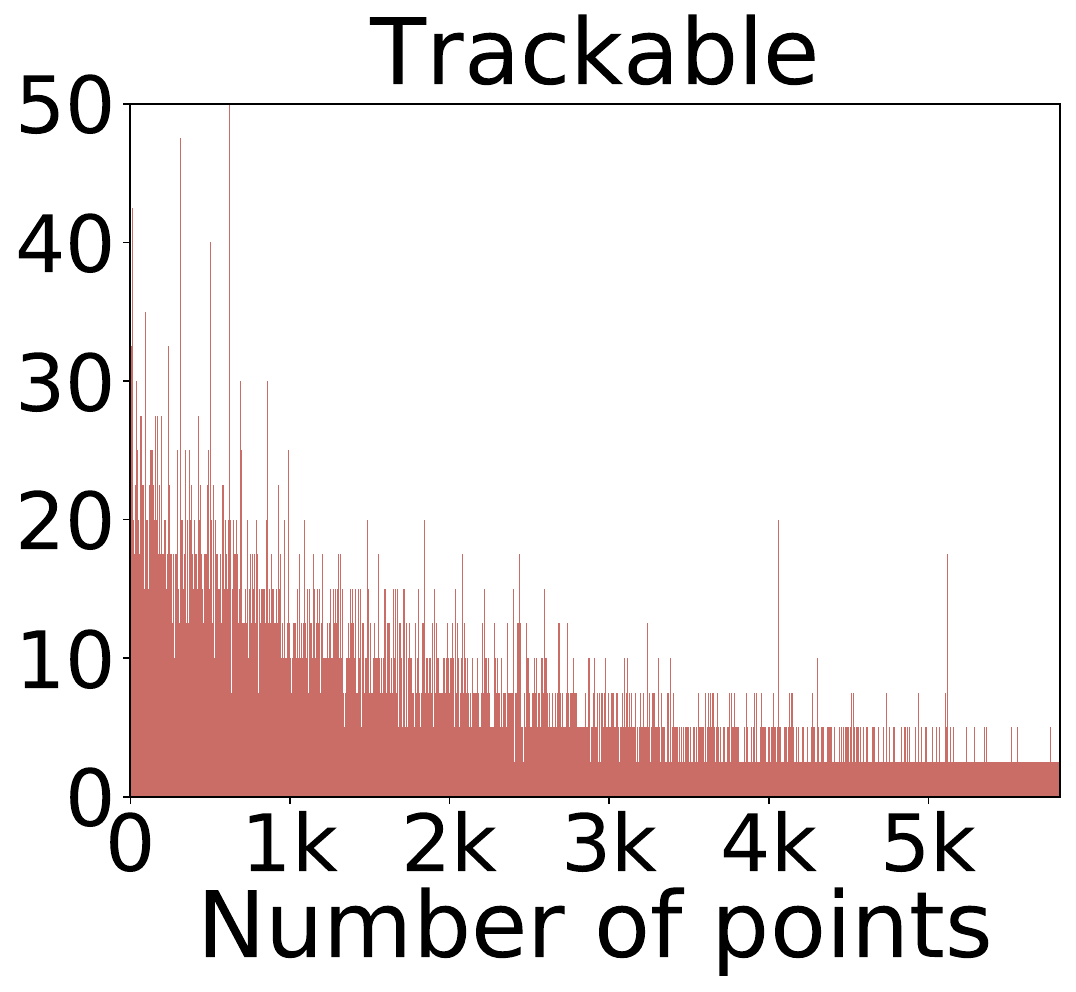}
\end{minipage}%
}%
\subfigure[]{
\begin{minipage}[t]{0.19\linewidth}
\centering
\includegraphics[width=3.2cm,height=2.6cm]{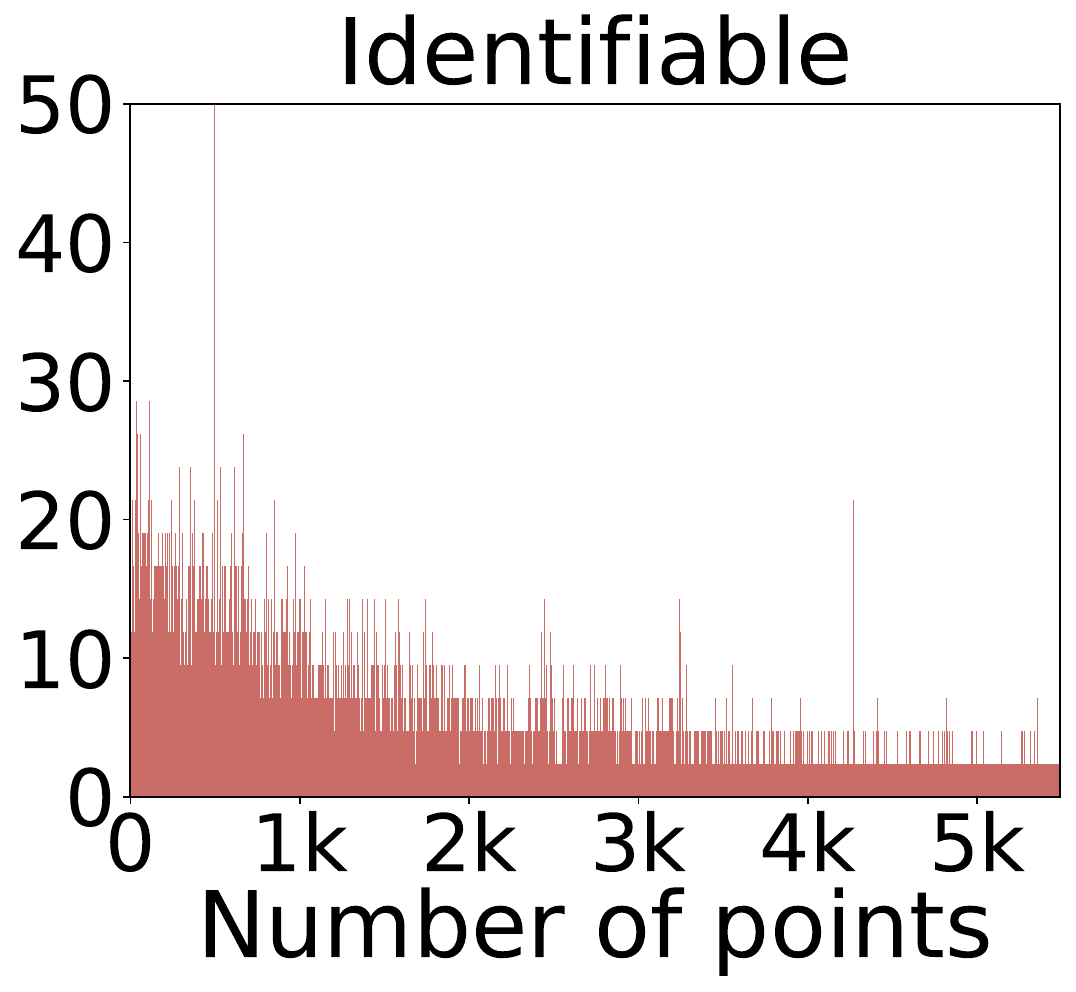}
\end{minipage}%
}%
\subfigure[]{
\begin{minipage}[t]{0.19\linewidth}
\centering
\includegraphics[width=3.2cm,height=2.6cm]{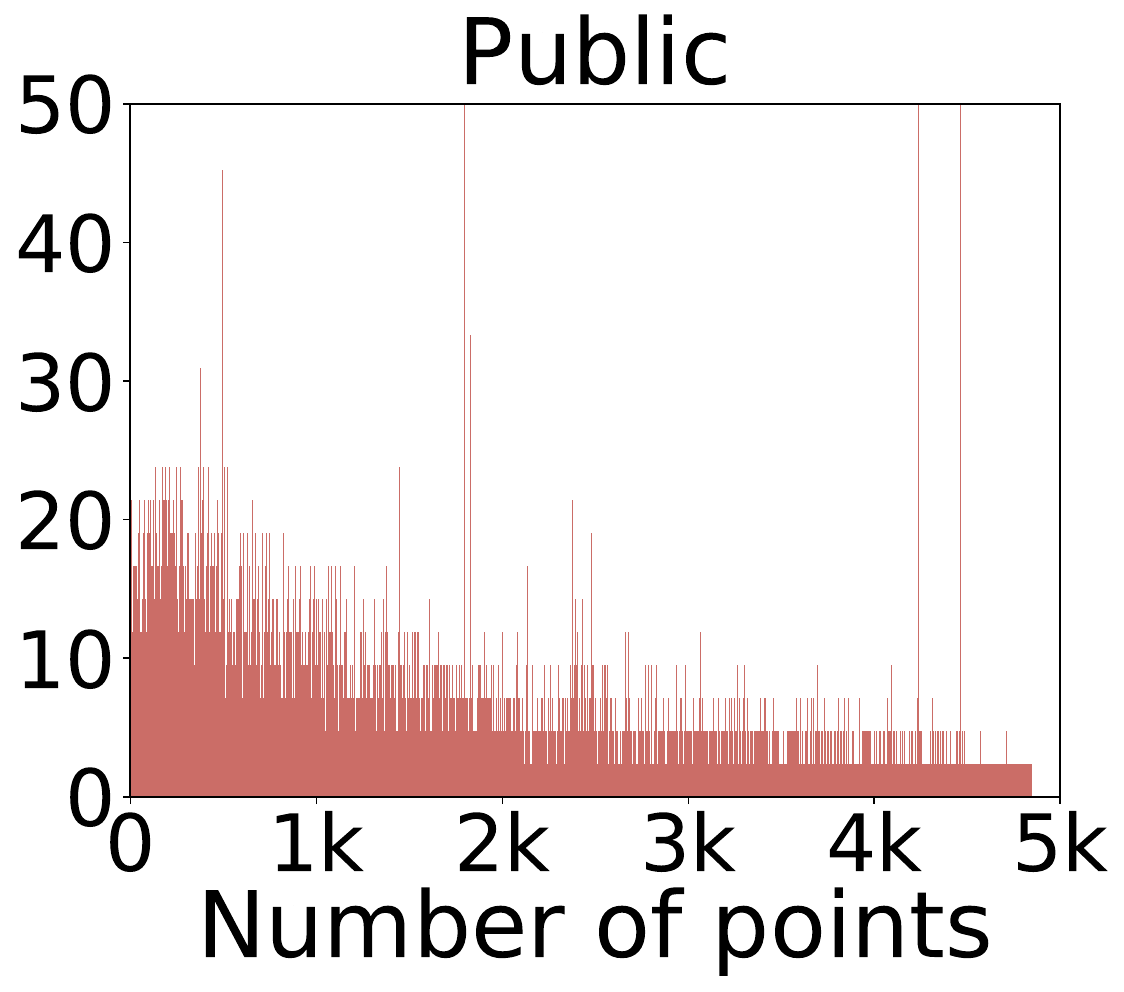}\end{minipage}%
}%
\subfigure[]{
\begin{minipage}[t]{0.19\linewidth}
\centering
\includegraphics[width=3.2cm,height=2.6cm]{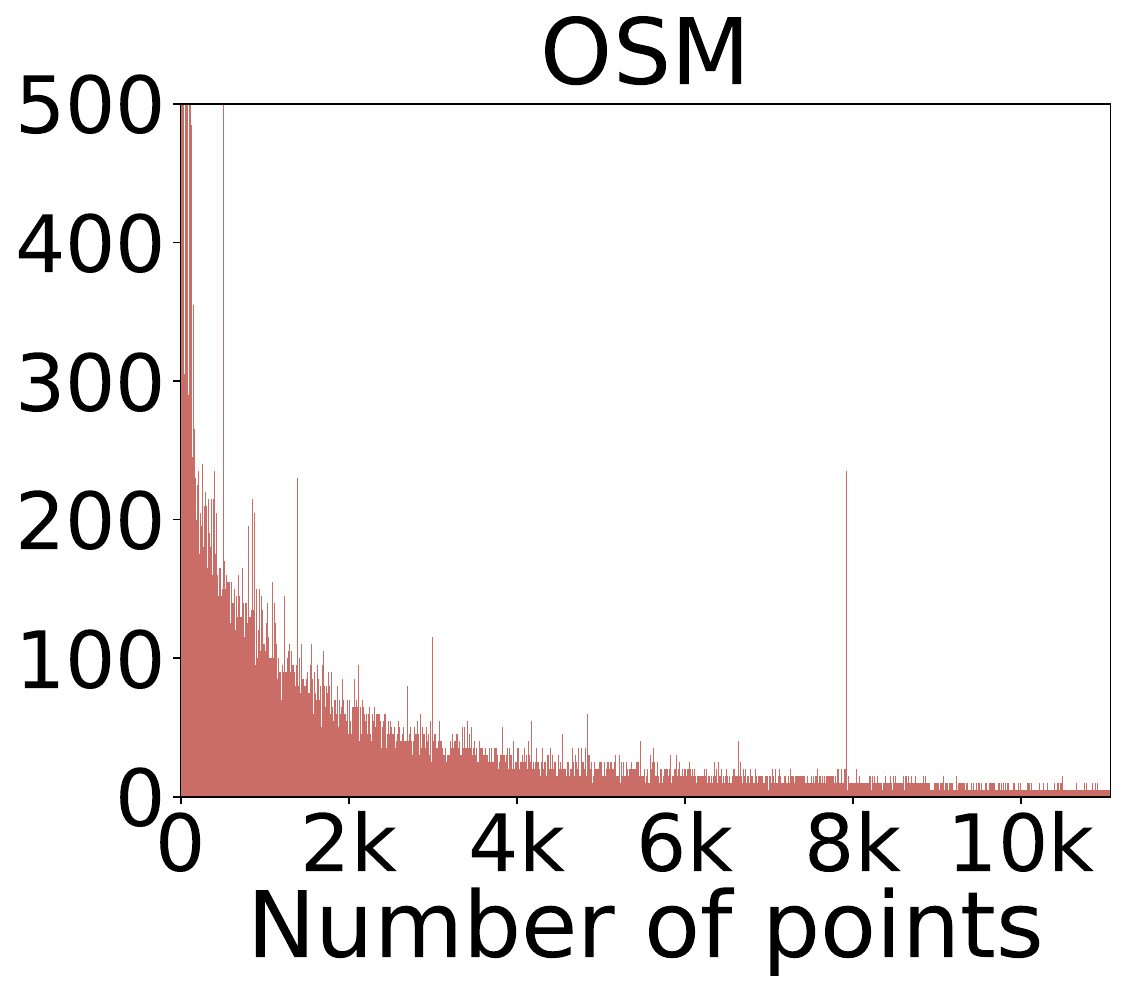}\end{minipage}%
}%
\centering
\caption{Dataset distributions in five data collections.}
\label{fig:distribution}
\end{figure*}

\begin{table*}[t]
\renewcommand{\arraystretch}{1.0}
\setlength\tabcolsep{2.0pt}
\caption{Statistics of spatial dataset graphs. }
\scalebox{1}{
\label{tab:graphSta}
\begin{tabular}{cccccclccccclccccclccccc}
\hline
\multirow{3}{*}{Collection}               & \multicolumn{11}{c}{Increasing $\delta$}                                                &  & \multicolumn{11}{c}{Increasing $\theta$}                              \\ \cline{2-12} \cline{14-24} 
& \multicolumn{5}{c}{Average Degree}      &  & \multicolumn{5}{c}{Number of Subgraphs} &  & \multicolumn{5}{c}{Average Degree}       &  & \multicolumn{5}{c}{Number of Subgraphs} \\ \cline{2-6} \cline{8-12} \cline{14-18} \cline{20-24} 
& 0    & 5     & 10     & 15     & 20     &  & 0       & 5     & 10    & 15    & 20    &  & 9      & 10     & 11     & 12    & 13    &  & 9     & 10     & 11    & 12    & 13     \\ \cline{1-6} \cline{8-12} \cline{14-18} \cline{20-24} 
\textsf{BTAA}       & 27.2 & 324.8 & 481.6  & 599.5  & 781.7  &  & 1759    & 66    & 36    & 21    & 20    &  & 1307.3 & 781.4  & 481.6  & 326.4 & 185.6 &  & 15    & 19     & 36    & 67    & 118    \\
\textsf{Trackable}    & 4.6  & 72.2  & 185.3  & 330.1  & 500.9  &  & 5600    & 449   & 180   & 100   & 69    &  & 1364.6 & 502.6  & 185.3  & 71.6  & 33.3  &  & 42    & 72     & 180   & 447   & 942    \\
\textsf{Identifiable} & 9.2  & 104.4 & 180.4  & 242.8  & 293.5  &  & 4299    & 594   & 262   & 149   & 94    &  & 554.3  & 293.3  & 180.6  & 104.3 & 60.0  &  & 36    & 96     & 262   & 594   & 1108   \\
\textsf{Public}       & 9.8  & 127.1 & 265.3  & 440.8  & 653.2  &  & 5079    & 292   & 127   & 87    & 63    &  & 1698.5 & 655.7  & 265.3  & 127.2 & 74.4  &  & 31    & 64     & 127   & 289   & 618    \\
\textsf{OSM}          & 32.2 & 443.1 & 1039.8 & 1894.4 & 2951.8 &  & 24228   & 904   & 358   & 215   & 123   &  & 7676.3 & 2955.2 & 1039.7 & 441.5 & 252.5 &  & 41    & 127    & 358   & 899   & 2260   \\ \cline{1-18} \cline{19-24} 
\end{tabular}
}
\end{table*}

\begin{itemize}[leftmargin=*]
\item[$\bullet$] Budget ($B$). It defines the total allowable budget for selecting datasets.
\item[$\bullet$] Distance threshold ($\delta$). It specifies the distance requirement for any pair of datasets to be considered directly connected.
\item[$\bullet$] Resolution ($\theta$). It determines the size of uniform cells in the grid. The larger the value of $\theta$, the smaller the cell size, and vice versa.
\item[$\bullet$] Scale ($m$). It specifies the scale of the dataset in each dataset collection used to verify the scalability of our algorithms.
\end{itemize} 

Firstly, we use the ratio of the total prices of all datasets in each data collection as the buyer's budget. Moreover, the parameter $\theta$ can be set based on the desired sampling distance in actual scenarios. Specifically, one degree of longitude or latitude is about 111km. If we divide the globe into a $2^{13} \times 2^{13}$ grid, then each cell’s size is about 5km ×2.5km. Furthermore, we vary the dataset scale from 0.2$\times |\mathcal{D}|$ to 1.0$\times |\mathcal{D}|$ to verify the scalability of our algorithms.
Finally, the $\delta$ can be set by dividing the required distance (e.g., commuting distance) in the actual situation by the cell's side length. 
It is worth noting that our study is orthogonal to the choice of pricing function. Changing pricing functions only impacts the calculation of prices and does not impact our search algorithm.
Here, the price of each spatial dataset is set to the number of covered cells~\cite{AzcoitiaIL2023,Miao2022,Pei2020}.

\noindent\textbf{Methods for Comparisons.}
We extend the state-of-the-art combinatorial algorithm on connected maximum coverage problem~\cite{Vandin2011} for comparison, called \textsf{CMC}. It constructs the BFS tree for each node  $v \in \mathcal{V}$ and produces feasible solutions by traversing all paths from $v$ to any node $u\in \mathcal{V}$, iteratively selecting the path that maximizes the average spatial coverage increment. Due to its low scalability (e.g., it can only produce results on reasonably small datasets), we implement two variants, namely \textsf{CMC+MC} and \textsf{CMC+MG}, which produce the feasible solution for each subgraph and return the best solution. The details of the two variants and our algorithms are shown below. 
\noindent

\begin{itemize}[leftmargin=*]

\item  \CAMS selects an arbitrary node to build a BFS tree and iteratively selects the path with the maximum average coverage (MC) per node on the path despite the marginal gain.
\item \CAMC selects an arbitrary node to build a BFS tree and iteratively selects the path with the maximum average marginal gain (MG) per node on the path.
\item \SG is our proposed dual-search algorithm (see Section~\ref{subsec:SG}).
\item \BGPA is our proposed dual-path search algorithm (see Section~\ref{subsec:BGPA}). 
\item \textsf{DPSA+BA} is our proposed dual-path search algorithm combined with the BFS-based acceleration (BA) strategy (see Section~\ref{subsec:accelerate}). 

\end{itemize}

Considering the expensive cost of building the dataset graph, we apply our proposed index-based acceleration strategy (see Section~\ref{sec:accelerate1}) to all baselines above as a preprocessing procedure.

\begin{figure*}[h!]
\setlength{\abovecaptionskip}{-0.1 cm}
\setlength{\belowcaptionskip}{-0.4 cm}
\centering
\includegraphics[width=18cm, height=3.5cm]{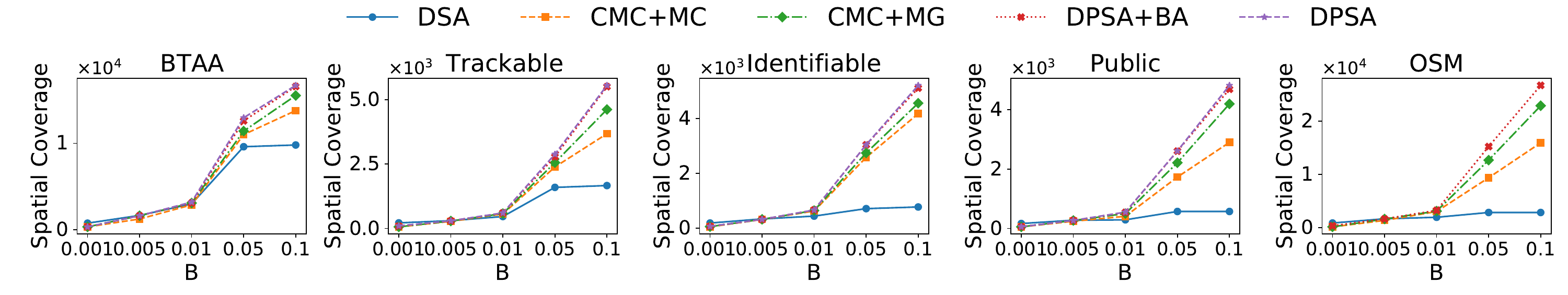}
  \caption{Comparison of the spatial coverage with increasing $B$. } 
  \label{fig:budget} 
\end{figure*}

\begin{figure*}[h!]
\setlength{\abovecaptionskip}{0 cm}
\setlength{\belowcaptionskip}{-0.4 cm}
\centering
\includegraphics[width=18cm, height=2.8cm]{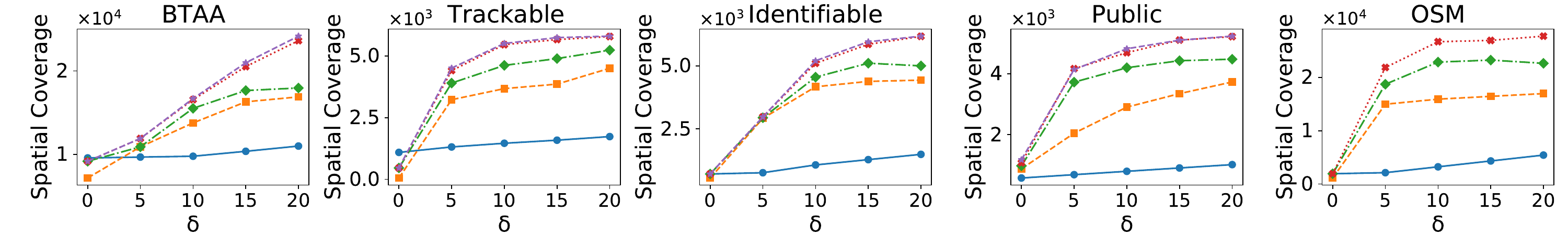}
 \caption{Comparison of the spatial coverage with increasing $\delta$.} 
\label{fig:delta}
\end{figure*}

\begin{figure*}[h!]
\setlength{\abovecaptionskip}{0 cm}
\setlength{\belowcaptionskip}{-0.5 cm}
\centering
\includegraphics[width=18cm, height=2.8cm]{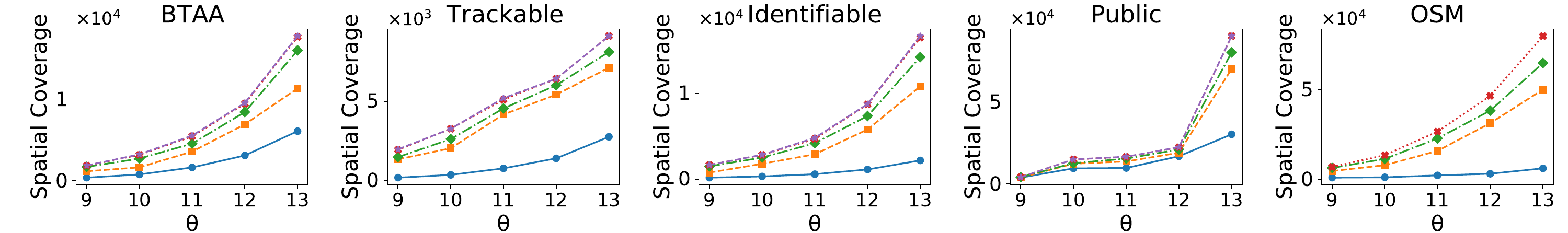}
\caption{Comparison of the spatial coverage with increasing $\theta$.} 
  \label{fig:reso} 
\end{figure*}

\begin{figure*}[ht]
\setlength{\abovecaptionskip}{-0.1 cm}
\setlength{\belowcaptionskip}{-0.3 cm}
\centering
\includegraphics[width=18cm, height=3.5cm]{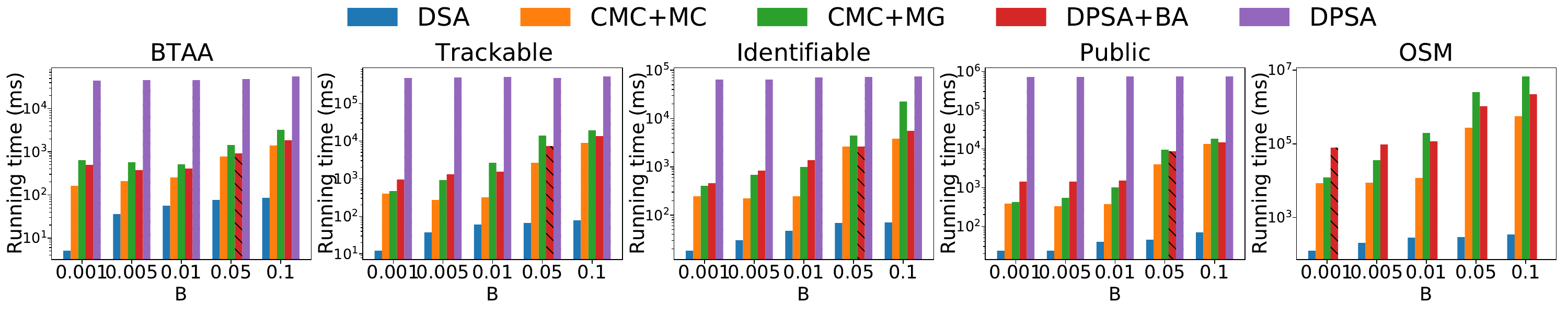}
  \caption{Comparison of the running time with the increase of $B$.} 
  \label{fig:budgetTime} 
\end{figure*}

\begin{figure*}[ht]
\setlength{\abovecaptionskip}{-0.1 cm}
\setlength{\belowcaptionskip}{-0.3 cm}
\centering
\includegraphics[width=18cm, height=3.2cm]{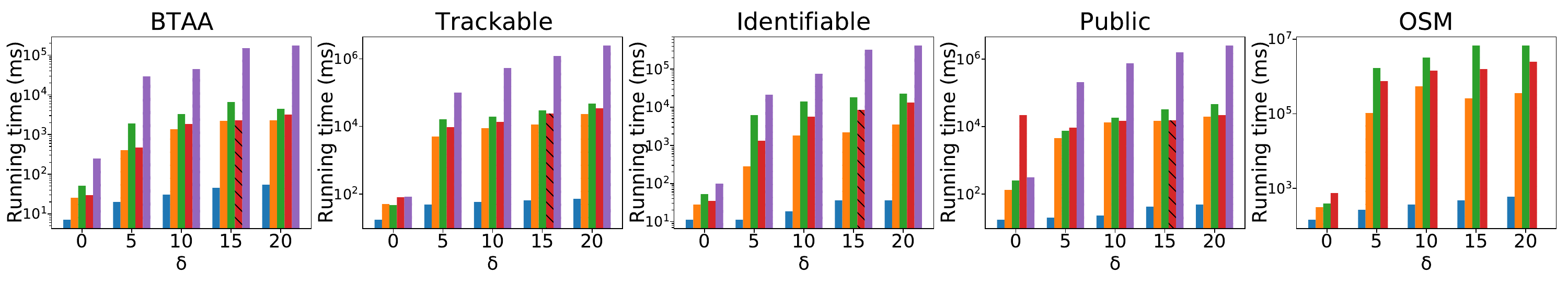}
\caption{Comparison of the running time with the increase of $\delta$.} 
\label{fig:deltaTime}
\end{figure*}

\begin{figure*}[h!]
\setlength{\abovecaptionskip}{-0.1 cm}
\setlength{\belowcaptionskip}{-0.3 cm}
\centering
\includegraphics[width=18cm, height=3.2cm]{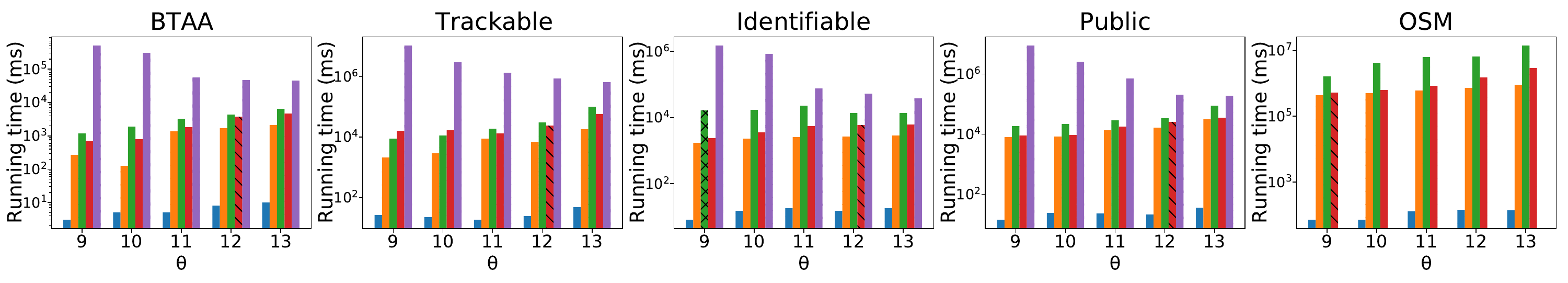}
\caption{Comparison of the running time with the increase of $\theta$.} 
\label{fig:resoTime}
\end{figure*}

\noindent\textbf{Evaluation Metrics.}
We evaluate the effectiveness based on the spatial coverage (i.e., the number of cells of the connected subgraph found by the algorithm).  Additionally, we assess the efficiency of all algorithms by comparing their running times.

\noindent\textbf{Environment.}
We run all experiments on a 10-core Intel(R) Xeon(R) Silver 4210 CPU @ 2.20GHZ processor, with 376G memory. All codes are implemented in Java 1.8, and any method that cannot finish within 24 hours will be terminated. 

\vspace{-0.4cm}
\subsection{Effectiveness Study}
\label{sec:effect}

\vspace{-0.2cm}
\noindent\textbf{Effectiveness Comparison with Different Budget $B$}.
Fig.~\ref{fig:budget} compares the performance of the five algorithms with different $B$. We can observe that \SG shows good performance when the budget ratio is 0.001 and 0.005. However, as the budget gradually increases, the performance of \SG becomes worse than the other four algorithms. This is because \SG is highly dependent on the selection of initial nodes. Additionally, since the graph we input is composed of multiple subgraphs, \SG may select the initial node belonging to one of the subgraphs and get stuck in a locally optimal solution.

Moreover, we observe that with increasing $B$, the solutions obtained by \BGPA cover at most 68\% times coverage than \CAMS and 21\% times coverage than \CAMC. This indicates the superior effectiveness of \BGPA over the previously proposed algorithm. Additionally, we can find that no matter how the budget changes, the solutions obtained by \BGPA and \BGPAF are within 1\% of the error. 
This demonstrates the effectiveness of our proposed BFS-based acceleration strategy for quickly finding the center. 
Notably, the experimental results of \textsf{DPSA} on 
\textsf{OSM} are not presented due to the algorithm's inability to complete within 24 hours. This is attributed to the time-consuming process of finding the center in \textsf{DPSA}.

\noindent\textbf{Effectiveness Comparison with Different Threshold $\delta$}.
Fig.~\ref{fig:delta} compares the performance of the five algorithms with different $\delta$. Firstly, as $\delta$ increases, the spatial coverage found by the five algorithms under the same budget gradually increases. This is because a higher $\delta$ means that more pairs of datasets satisfy the directly connected relation, resulting in a higher average degree in the graph, as shown in Table~\ref{tab:graphSta}. Thus, the spatial coverage in the solution found by the algorithm also increases. Additionally, we can observe that \BGPA consistently demonstrates robust performance irrespective of $\delta$ variations, achieving at most 63\% and 35\% times spatial coverage against \CAMC and \CAMS. This indicates that our algorithm is effective compared to the two variants. 

\noindent\textbf{Effectiveness Comparison with Different Resolution $\theta$}.
Fig.~\ref{fig:reso} compares the performance of the five algorithms with different $\theta$. We can see that as the $\theta$ increases, the spatial coverage is gradually increasing. This can be explained by the fact that in the stage of generating the cell-based dataset, we divide the entire space into a $2^\theta \times 2^\theta$ grid and generate the cell-based dataset for each raw dataset. Thus, the spatial coverage of each cell-based dataset increases as $\theta$ increases, leading to larger spatial coverage found by the five algorithms. 
Additionally, \BGPA and \BGPAF can achieve at most 58\% and 38\% times spatial coverage against \CAMC and \CAMS, respectively, which also verifies the effectiveness of our proposed algorithms.

\subsection{Efficiency Study}
\label{sec:efficient}
\noindent\textbf{Efficiency Comparison with Different Budget $B$}.
Fig.~\ref{fig:budgetTime} compares the running time of the five algorithms with different $B$, with execution time measured in milliseconds (ms). We have five observations:

\noindent(1) \textsf{DSA} always has the shortest running time compared to the other four algorithms, completing the search in a few milliseconds. Combined with the experimental results shown in Figs.~\ref{fig:budget} and~\ref{fig:budgetTime}, we can see that \textsf{DSA} offers the best trade-off in terms of speed and accuracy when the budget ratio does not exceed 0.005. 

\noindent(2) \textsf{DPSA+BA} is at least an order of magnitude faster than \textsf{DPSA}, yet it can still find solutions that are almost as good as \textsf{DPSA} (less than 1\% error). This demonstrates the effectiveness of our strategy for quickly finding the center in practical applications.

\noindent(3) Although \BGPA can find a better solution when the budget ratio exceeds 0.05, it exhibits significantly slower performance, typically by one or two orders of magnitude, compared to \CAMS and \CAMC. 
This is because the process of finding the center node takes significant time in \BGPA.
Additionally, \BGPAF can achieve up to at most  89\% times speedups over \CAMC, providing compelling evidence to support the proposition that utilizing the BFS-based acceleration strategy for expedited center node identification is highly efficient.

\noindent(4) Regardless of how the budget varies, the running time of \CAMS is always faster than that of \CAMC and \BGPA. This is because \CAMS only needs to add the dataset with the maximum coverage size to the result set in each iteration without computing the coverage increment. However, this acceleration strategy also results in the solution of \CAMS with less spatial coverage than \CAMC and \BGPAF.

\noindent(5) When $B$ is small, \CAMC runs slightly faster than \BGPA. However, as $B$ increases, \CAMC takes longer execution times compared to \BGPA.
This is because when $B$ is small, the number of affordable datasets is small (i.e., the algorithm has fewer iterations), making the running time of finding the center more important in the overall algorithm. \CAMC does not need to find a center, so the algorithm is slightly faster with a smaller $B$. However, as $B$ increases, the number of iterations also increases, and \CAMC needs to traverse more paths than \BGPA, resulting in a longer running time for \CAMC. 
Furthermore, combined with Fig.~\ref{fig:budget}, we can see that the result set found by \CAMC is worse than \BGPAF, which indicates that the two maximization strategies adopted in our algorithm are effective and robust.

\noindent\textbf{Efficiency Comparison with Different Threshold $\delta$}.
As shown in Fig.~\ref{fig:deltaTime}, we can observe that as $\delta$ increases, the running time of the five algorithms gradually increases. Among them, \SG has the shortest running time and the smallest increase, while \BGPA has the maximum increase. This is because the time taken to find the center in \BGPA increases significantly as $\delta$ increases. Another interesting phenomenon is that apart from \BGPA and \SG, \CAMC has the slowest running time as $\delta$ increases, followed by \BGPAF, and \CAMS is the fastest. This is because \CAMC builds the path from the root node to all other nodes at each iteration, while \BGPAF only considers the path from the root node to the leaf node. Combined with Fig.~\ref{fig:budget}, we can see that although the running time of \CAMS is shorter than that of \BGPAF and \CAMC, the solution found is less effective because it does not consider the marginal gain.

\vspace{-0.2em}
\noindent\textbf{Efficiency Comparison with Different Resolution $\theta$}.
Fig.~\ref{fig:resoTime} compares the running time of the five algorithms with different $\theta$. We have an interesting observation from this figure. As the $\theta$ increases, only the running time of \BGPA decreases, while the running time of the other four algorithms increases gradually. This is because as the $\theta$ increases, the entire space is divided into finer granularity, resulting in a higher spatial coverage in each cell-based dataset. Furthermore, the finer granularity of each dataset contributes to an increase in spatial dataset distance between pairs of datasets, causing more node pairs that cannot be directly connected. This, in turn, leads to an increase in the number of subgraphs and a decrease in the average degree of subgraphs, as shown in Table~\ref{tab:graphSta}. Thus, the time taken to find the center in \BGPA decreases significantly, resulting in a decrease in the overall running time of \BGPA.

\subsection{Scalability}
\label{sec:scale}
The objective of this experiment is to evaluate the scalability of our algorithms and two comparison algorithms when we increase the dataset size. Fig.~\ref{fig:scale} shows the results of five algorithms’ spatial coverage and running time, respectively, when we vary the dataset size from 0.2$\times |\mathcal{D}|$ to 1.0$\times |\mathcal{D}|$ and set the other parameters as defaults. From the left of Fig.~\ref{fig:scale}, we can observe that the spatial coverage of five algorithms slightly increases as the dataset scale increases. This is because, as the dataset scale increases, each dataset has more potential connectable datasets, leading to a higher degree of nodes in the dataset graph. Therefore, under the given budget, the algorithm is able to select more datasets that contribute to improving spatial coverage, thus enhancing the overall coverage effect. However, no matter how the size of the dataset increases, our \BGPA and \BGPAF can still maintain a relatively larger spatial coverage than other comparison algorithms, which exhibit good scalability characteristics.

From the right of Fig.~\ref{fig:scale}, we can see that the running time of five algorithms slightly increases as the dataset scale increases. This is because the algorithms need more time to process the increased number of connectable nodes. Among the five algorithms, \SG consistently exhibits the shortest running time, completing the search in just a few milliseconds, which demonstrates its strong scalability. However, due to the relatively large budget set by default in this scalability experiment, the approximation ratio of \SG decreases, resulting in its spatial coverage performance being worse than that of the other comparison algorithms. This observation is consistent with our theoretical analysis. In addition to \SG, we can observe that \BGPAF has a faster search speed than other methods, which indicates that \BGPAF can efficiently handle larger datasets and maintain high performance even as the size of the data increases.

\begin{figure*}
\setlength{\abovecaptionskip}{0 cm} 
\setlength{\belowcaptionskip}{-0.3 cm} 
\subfigure[]{
\begin{minipage}[t]{0.50\textwidth}
\centering
\includegraphics[width=8.9cm]{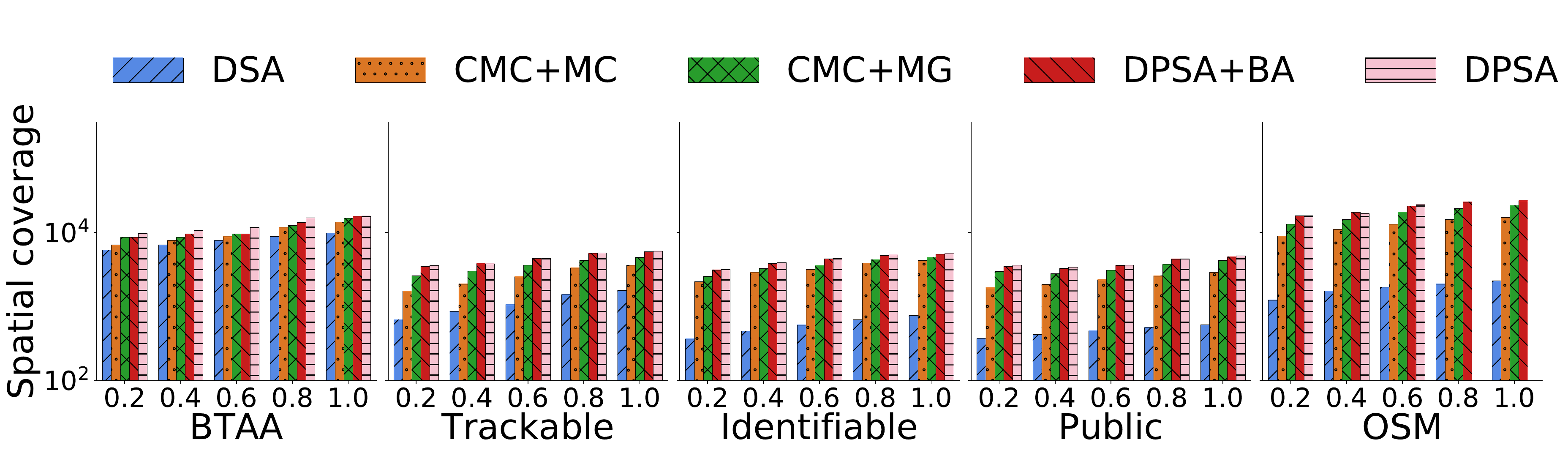}
\end{minipage}%
}%
\subfigure[]{
\begin{minipage}[t]{0.50\textwidth}
\centering
\includegraphics[width=8.9cm]{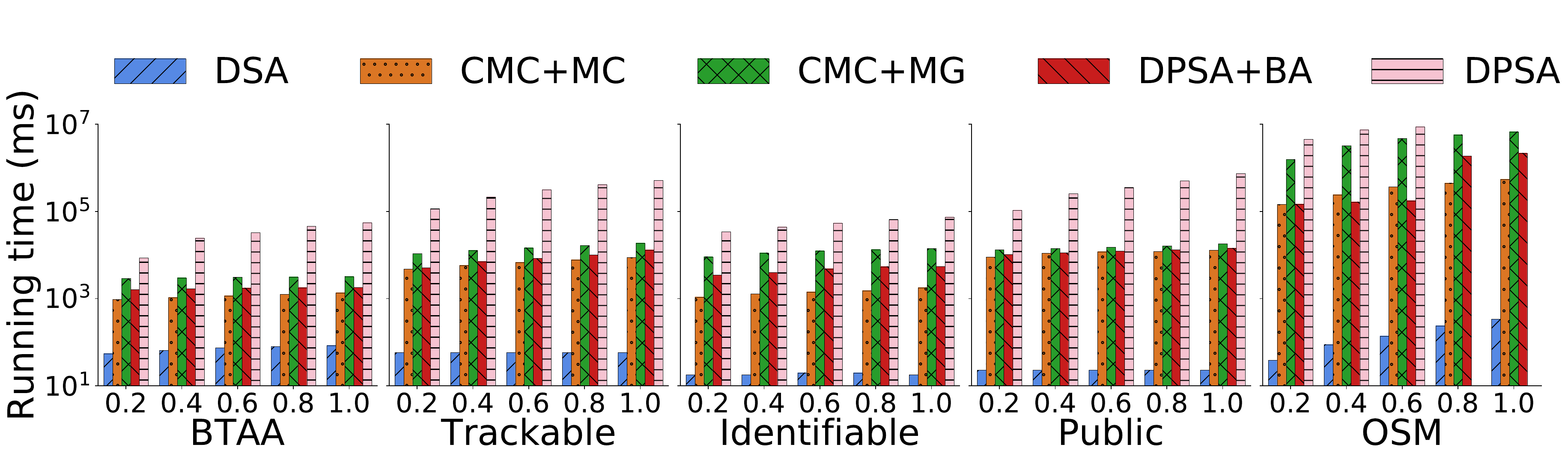}
\end{minipage}%
}%
\centering
\caption{Comparison of five algorithms' spatial coverage and running time with dataset scale $m$ increasing.}\label{fig:scale}
\end{figure*}

\begin{figure*}
\setlength{\abovecaptionskip}{0 cm}
\setlength{\belowcaptionskip}{-0.5 cm}
\centering
\begin{minipage}[t]{0.33\linewidth}
\centering
\subfigure[Raw datasets]{
\includegraphics[width=4cm,height=3.3cm]{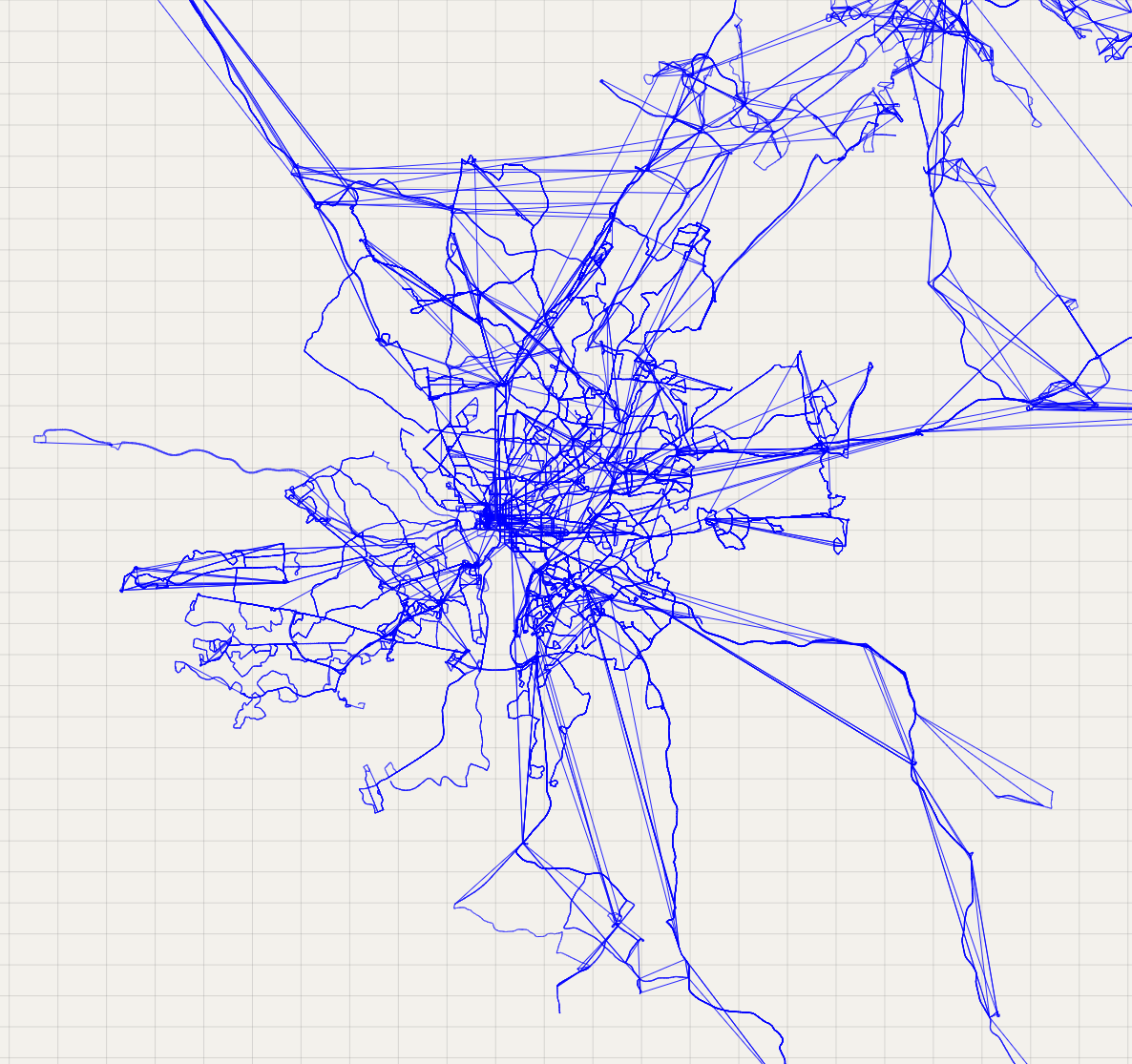}
}
\end{minipage}%
\begin{minipage}[t]{0.33\linewidth}
\centering
\subfigure[Results of \SG]{
\includegraphics[width=4cm,height=3.3cm]{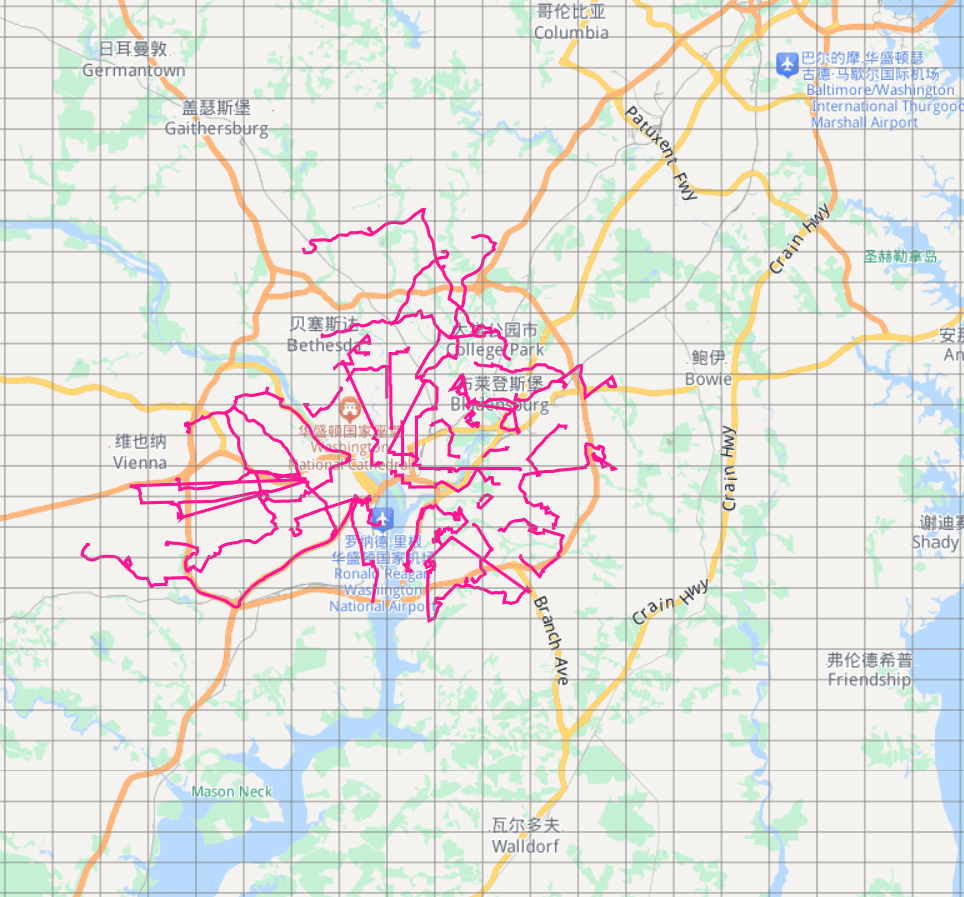}
}
\end{minipage}%
\begin{minipage}[t]{0.33\linewidth}
\centering
\subfigure[Results of \CAMS]{
\includegraphics[width=4cm,height=3.3cm]{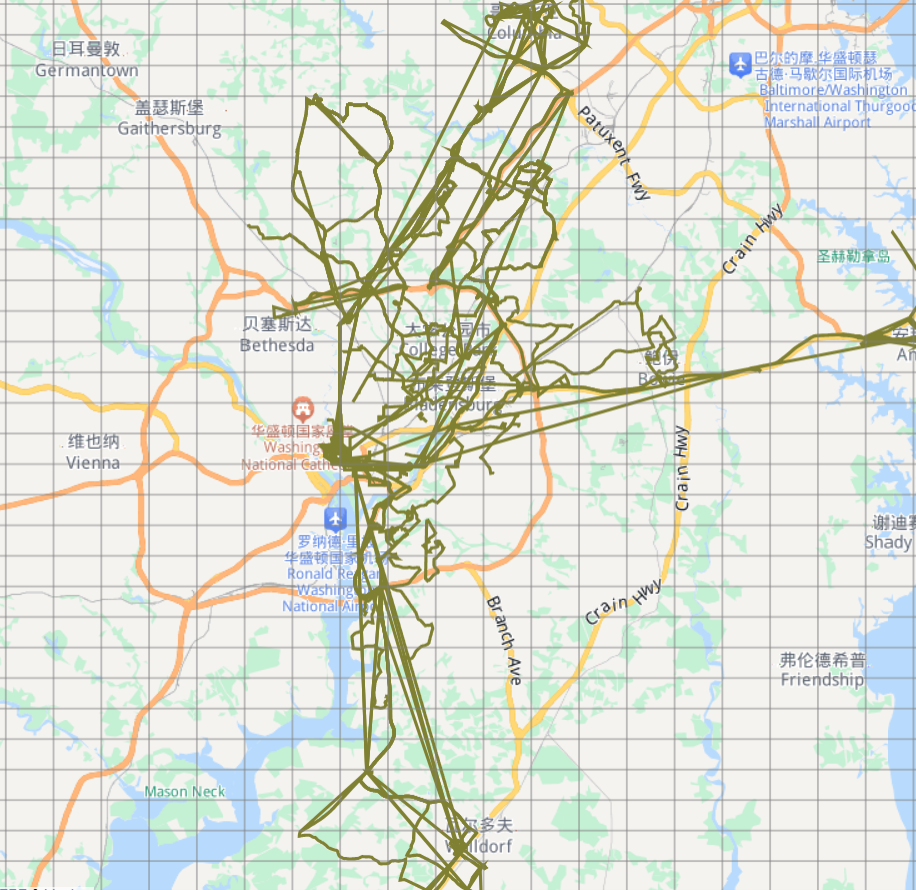}
}
\end{minipage}

\begin{minipage}[t]{0.33\linewidth}
\centering
\subfigure[Results of \CAMC]{
\includegraphics[width=4cm,height=3.3cm]{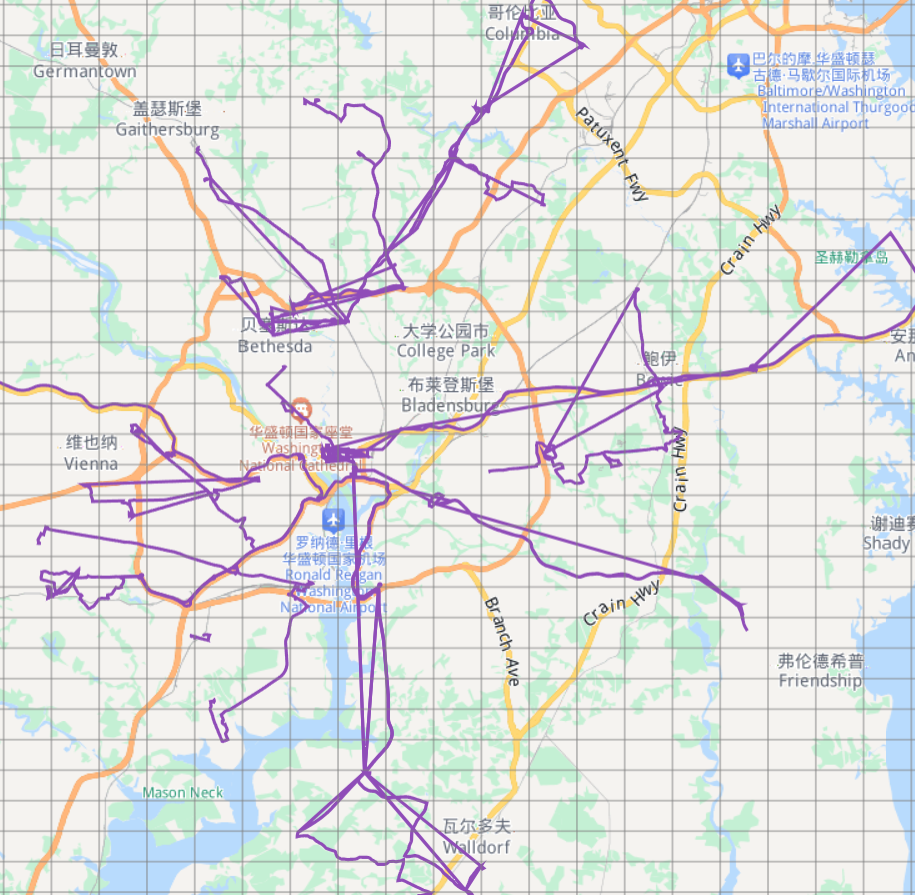}
}
\end{minipage}%
\begin{minipage}[t]{0.33\linewidth}
\centering
\subfigure[Results of \BGPAF]{
\includegraphics[width=4cm,height=3.3cm]{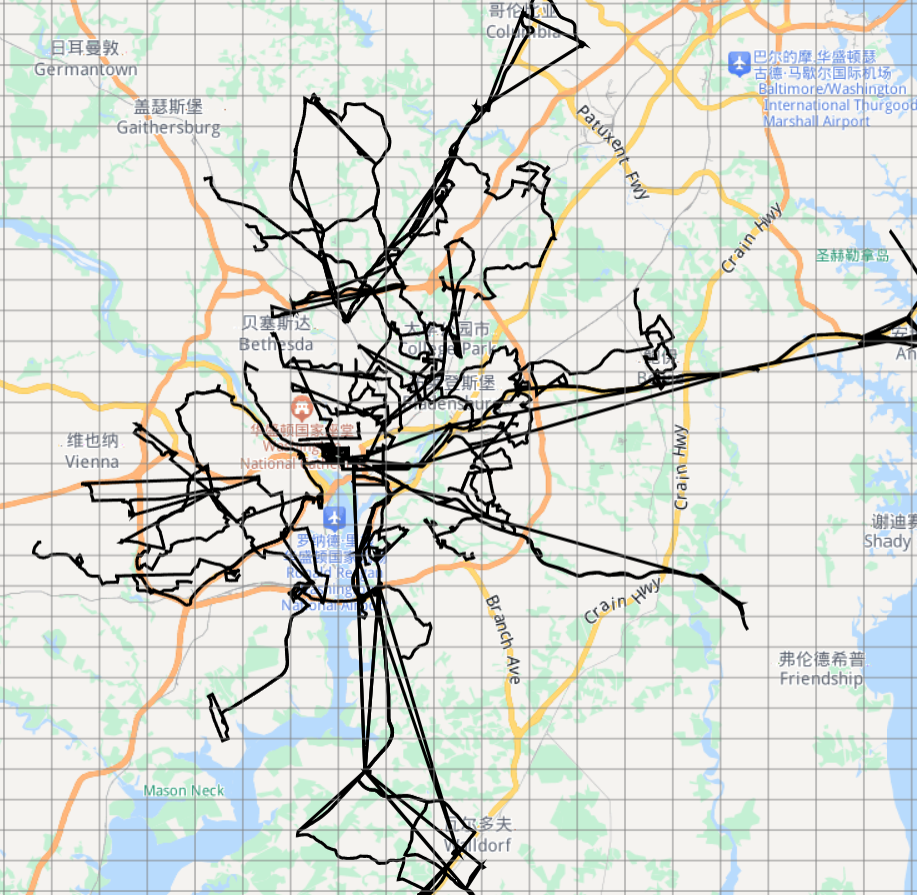}
}
\end{minipage}%
\begin{minipage}[t]{0.33\linewidth}
\centering
\subfigure[Results of \BGPA]{
\includegraphics[width=4cm,height=3.3cm]{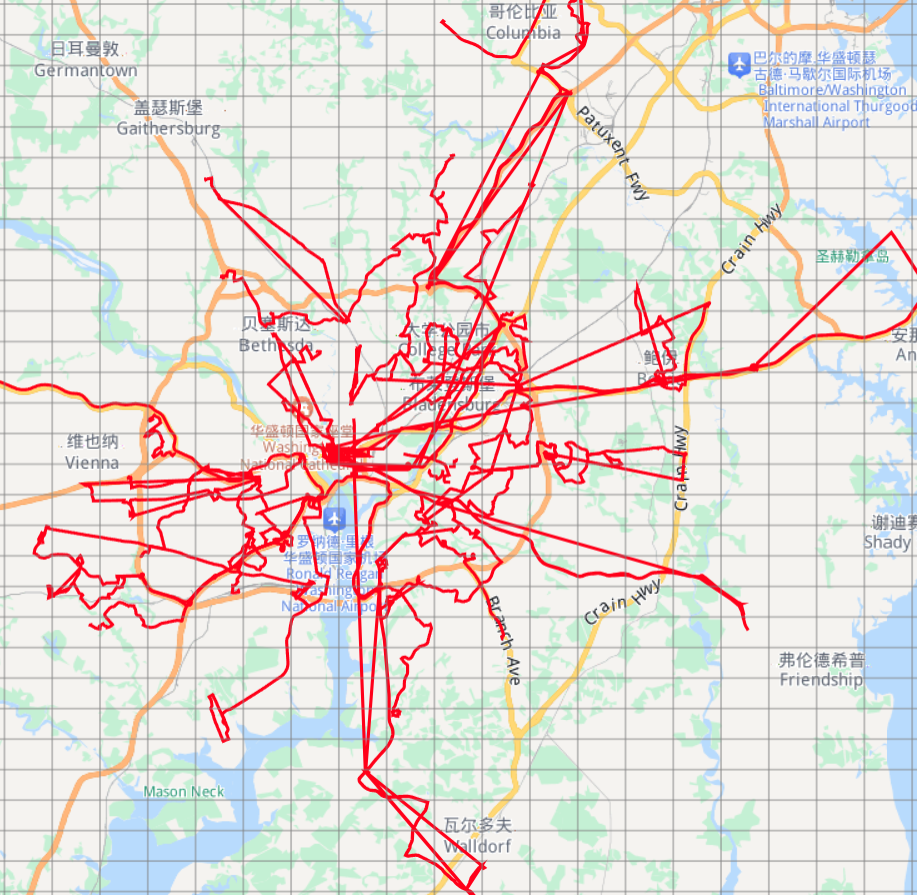}
}
\end{minipage}

\setlength{\belowcaptionskip}{-1pt}
\caption{Case studies on private trajectories.}
\label{fig:visual}
\vspace{-0.3cm}
\end{figure*}

\subsection{Visual Analysis}
\label{sec:case}
For visualization purposes, we show 2000 individual trajectories located at Washington from \textsf{BTAA}, as shown in Fig.~\ref{fig:visual}(a). The solutions found by five algorithms with the default settings are shown in Figs.~\ref{fig:visual}(b)--(f). 
As shown in Fig.~\ref{fig:visual}(b), We can clearly see that the datasets obtained by \SG exhibit lower coverage, whereas the datasets obtained by \CAMC demonstrate higher coverage compared to \CAMS and \SG when the budget is 0,1. This reveals that the performance of \SG becomes worse than other algorithms when the budget is larger due to the decrease in the approximation ratio.
Additionally, the datasets found by \CAMS are only distributed on one side of the entire space, which has a one-sidedness. In contrast, we observe that the datasets obtained by \BGPA and \BGPAF are distributed over the entire space. This shows that our algorithm has better spatial coverage compared to baselines, which validates the effectiveness of our proposed algorithms. For data buyers who do not have the budget to acquire all datasets, this solution is a good representation of the dataset distribution in the entire data marketplace.
\end{sloppypar}

\section{Conclusions}
\label{sec:conclude}

In this paper, we proposed and studied the \MCPBC problem in the spatial data marketplace, which aims to maximize spatial coverage while maintaining spatial connectivity under a limited budget. We proved the NP-hardness of this problem. To address the \MCPBC, we proposed two approximation algorithms called \SG and \BGPA with theoretical guarantees and time
complexity analysis, followed by two accelerating strategies to further improve the efficiency of \BGPA. We conducted extensive experiments to demonstrate that our algorithms are highly efficient and effective. In the future, we will explore fairness-aware spatial data acquisition and integration.

\begin{acknowledgement}
This work was supported by the National Key R\&D Program of China (Grant No. 2023YFB4503600), the National Natural Science Foundation of China (Grant Nos. 62202338 and 62372337), and the Key R\&D Program of Hubei Province (Grant No. 2023BAB081).
\end{acknowledgement}

\begin{competinginterest}
The authors declare that they have no competing interests or financial conflicts to disclose.
\end{competinginterest}

\bibliographystyle{fcs}
\bibliography{reference-1.bib, reference-2.bib}

\begin{biography}{FCS-author1-Yang}
{Wenzhe Yang} received the BE degree and ME degree in computer science and technology from Jilin University. She is currently working toward the PhD degree with the School of Computer Science at Wuhan University. Her research interests mainly include data sharing, dataset search, and data market.
\end{biography}

\vspace{1cm}
\begin{biography}{FCS-author2-Huang}
Shixun Huang received the Bachelor degree from Nanjing University, Master Degree from University of Melbourne and Ph.D. degree from RMIT University. He is currently a Lecturer at University of Wollongong. His current research
interests span across combinatorial optimization and mining in graph data.
\end{biography}

\begin{biography}{FCS-author3-Wang}
{Sheng Wang} received the BE degree in information security, ME degree in computer technology from Nanjing University of Aeronautics and Astronautics, China in 2013 and 2016, and Ph.D. from RMIT University in 2019. He is a professor at the School of Computer Science, Wuhan University. His research interests mainly include spatial databases. He has published full research papers on top database and information systems venues as the first author, such as TKDE, SIGMOD, PVLDB, and ICDE.
\end{biography}

\begin{biography}{FCS-author4-Peng}
{Zhiyong Peng} received the BSc degree from Wuhan University, in 1985, the MEng degree from the Changsha Institute of Technology of China, in 1988, and the PhD degree from the Kyoto University of Japan, in 1995. He is a professor at the School of Computer Science, Wuhan University.  His research interests mainly include complex data management, web data management, and trusted data management. 
\end{biography}

\end{document}